\documentclass[11pt]{article}

\usepackage[T1]{fontenc}
\usepackage{lmodern}

\usepackage{array}
\usepackage{amsmath, amssymb}
\usepackage{yhmath}
\usepackage{amsthm}
\usepackage{proof}
\usepackage{enumitem}
\usepackage{algpseudocodex}
\usepackage{algorithm,setspace}

\usepackage{mathtools}
\usepackage{fullpage}
\usepackage{scalerel}
\usepackage{hhline}
\usepackage{bbm}
\usepackage[normalem]{ulem}
\usepackage{threeparttable}
\usepackage{booktabs}
\usepackage[pdfstartview=FitH,pdfpagemode=UseNone,colorlinks=true,citecolor=blue,linkcolor=blue]{hyperref}
\usepackage{centernot}
\usepackage[
	lambda,
	operators,
	advantage,
	sets,
	adversary,
	landau,
	probability,
	notions,
	logic,
	ff,
	mm,
	primitives,
	events,
	complexity,
	asymptotics,
	keys]{cryptocode}
\usepackage[T1]{fontenc}
\usepackage{graphicx}
\usepackage{framed}
\usepackage{mdframed}
\usepackage{thmtools}
\usepackage{thm-restate}
\usepackage[noabbrev]{cleveref}
\usepackage{verbatim}
\floatname{algorithm}{Protocol}
\crefname{algorithm}{protcol}{protocol}
\Crefname{algorithm}{Protocol}{Protocol}

\usepackage[utf8]{inputenc}
\usepackage{csquotes}
\usepackage{microtype}
\usepackage{multirow}
\usetikzlibrary{decorations.pathreplacing,matrix,shapes,arrows,positioning,chains,calc}

\tikzset{
    ncbar angle/.initial=90,
    ncbar/.style={
        to path=(\tikztostart)
        -- ($(\tikztostart)!#1!\pgfkeysvalueof{/tikz/ncbar angle}:(\tikztotarget)$)
        -- ($(\tikztotarget)!($(\tikztostart)!#1!\pgfkeysvalueof{/tikz/ncbar angle}:(\tikztotarget)$)!\pgfkeysvalueof{/tikz/ncbar angle}:(\tikztostart)$)
        -- (\tikztotarget)
    },
    ncbar/.default=0.5cm,
}

\tikzset{square left brace/.style={ncbar=0.5cm}}
\tikzset{square right brace/.style={ncbar=-0.5cm}}

\tikzset{round left brace/.style={ncbar=0.5cm,out=-115,in=115}}
\tikzset{round right brace/.style={ncbar=0.5cm,out=-65,in=65}}

\usepackage{subcaption}
\usepackage{wrapfig}
\usepackage{xcolor}
\usepackage{xspace}

\usepackage[normalem]{ulem}
\usepackage{bm}

\newif\ifFOCS
\ifFOCS
\Crefname{figure}{Fig.}{Fig.}
\else
\usepackage[colorinlistoftodos,textsize=small]{todonotes}
\fi

\theoremstyle{remark}
\newtheorem{remark}{Remark}

\theoremstyle{definition}
\newtheorem{definition}{Definition}

\theoremstyle{plain}
\newtheorem{claim}{Claim}

\newtheorem{lemma}{Lemma}
\newtheorem{theorem}{Theorem}
\newtheorem{proposition}{Proposition}
\newtheorem{corollary}{Corollary}

\newtheorem{fact}{Fact}

\newcommand{\newuser}[3]{%
  \expandafter\newcommand\csname todo#1\endcsname[2][]{%
    \quitvmode
    \texorpdfstring{\todo[inline,color=#3,##1]{\textbf{#2:} ##2}\xspace}{(TODO: #2: ##2)}%
  }
}

\newcommand\numberthis{\addtocounter{equation}{1}\tag{\theequation}}

\definecolor{paramnotebg}{RGB}{255,248,225}
\definecolor{paramnoteborder}{RGB}{166,91,22}

\definecolor{paramfixbg}{RGB}{232,246,239}
\definecolor{paramfixborder}{RGB}{31,112,68}

\DeclarePairedDelimiter{\lrag}{\lag}{\rag}

\newcommand{\polylog}[1]{\mathrm{polylog}{#1}\xspace}
\renewcommand{\poly}[1]{{\mathrm{poly}{#1}}}

\newcommand{\tO}{\widetilde{O}}

\mathchardef\mhyphen="2D

\let\lag\langle
\let\rag\rangle

\crefname{conjecture}{conjecture}{Conjectures}
\crefname{question}{question}{Questions}
\crefname{claim}{claim}{Claims}
\crefname{fact}{fact}{Facts}

\makeatletter
\DeclareOldFontCommand{\bf}{\normalfont\bfseries}{\mathbf}
\DeclareOldFontCommand{\it}{\normalfont\itshape}{\mathit}
\DeclareOldFontCommand{\rm}{\normalfont\rmfamily}{\mathrm}
\DeclareOldFontCommand{\sc}{\normalfont\scshape}{\@nomath\sc}
\DeclareOldFontCommand{\sf}{\normalfont\sffamily}{\mathsf}
\DeclareOldFontCommand{\sl}{\normalfont\slshape}{\@nomath\sl}
\DeclareOldFontCommand{\tt}{\normalfont\ttfamily}{\mathtt}
\makeatother
\def\gkeywords{}
\def\gabstract{}

\newcommand\makealltitles{
\ifnum\style=\sigalternate
    {
        \begin{abstract}{\gabstract}\end{abstract} \maketitle%
    }
\else
  {\maketitle%
  \ifnum\style=\lncs
    {\begin{abstract}{\gabstract \keywords{\gkeywords}}\end{abstract}}
  \else
    {\begin{abstract}{\gabstract}\end{abstract}
      \ifnum\style=\ieeetr
      \begin{IEEEkeywords}
        {\gkeywords}
      \end{IEEEkeywords}
      \fi
    }
  \fi
  }
\fi
}

\mathchardef\mhyphen="2D

\newcommand{\ie}{{\it i.e.}\xspace}
\newcommand{\ignore}[1]{}

\usepackage{xcolor}
\newcommand{\printparam}[1]{_{\color{gray}[#1]}}
\newcommand{\prot}[4][]{(\cP{#2#3}, \cV{#2#4})\if\relax\detokenize{#1}\relax\else{\printparam{#1}}\fi\xspace}
\newcommand{\outputprotUU}[3][]{\langle\cP{#2}, \cV{#3}\rangle\if\relax\detokenize{#1}\relax\else{\printparam{#1}}\fi\xspace}

\newcommand{\cL}{\mathcal{L}}
\newcommand{\cM}{\mathcal{M}}

\newcommand{\cS}{\mathcal{S}}
\newcommand{\NA}{---}

\let\oldprocedure\procedure
\renewcommand\procedure[3][]{\hspace{0pt}\oldprocedure[#1]{#2}{#3}}

\NewDocumentCommand\inputanon{m}{
\ifnum\paperversion=\anonymous \input{#1} \fi
}

\NewDocumentCommand\inputcam{m}{
\ifnum\paperversion=\cameraready \input{#1} \fi
}

\NewDocumentCommand\inputfull{m}{
\ifnum\paperversion=\fullversion \input{#1} \fi
}

\newcommand{\cB}{\mathcal{B}}

\newcommand{\cU}{{\ensuremath{\mathcal{U}}\xspace}}

\newcommand{\cP}{{\ensuremath{\mathsf{P}}\xspace}}
\newcommand{\cV}{{\ensuremath{\mathsf{V}}\xspace}}

\newcommand{\GF}{\mathbb{GF}}

\newcommand{\cksum}{\ensuremath{\mathsf{cksum}}\xspace}

\newcommand{\Deltac}{\ensuremath{\Delta_c}\xspace}
\newcommand{\pval}{\textsf{PVAL}\xspace}
\newcommand{\pvalF}{\textsf{PVAL}_\FF\xspace}

\newcommand{\eps}{\ensuremath{\epsilon}\xspace}

\newcommand{\GKR}{\ensuremath{\textsf{GKR}}\xspace}
\newcommand{\Batch}{\ensuremath{\textsf{Batch}}\xspace}
\newcommand{\SmallBatch}{\ensuremath{\textsf{SmallBatch}}\xspace}

\newcommand{\ADD}{\ensuremath{\textsf{ADD}}\xspace}
\newcommand{\MULT}{\ensuremath{\textsf{MULT}}\xspace}

\renewcommand{\RR}{\ensuremath{\textsf{RR}}\xspace}

\newcommand{\tstart}{{\ensuremath{\mathsf{start}}\xspace}}
\newcommand{\tend}{{\ensuremath{\mathsf{end}}\xspace}}

\newcommand{\smallnwarrow}{\mathchoice
  {\scalebox{0.7}{$\nwarrow$}}
  {\scalebox{0.7}{$\nwarrow$}}
  {\scalebox{0.6}{$\nwarrow$}}
  {\scalebox{0.6}{$\nwarrow$}}
}

\newcommand{\texp}{{\ensuremath{\smallnwarrow}}}
\newcommand{\tred}{{\ensuremath{\downarrow}}}

\newcommand{\pvalT}{{\ensuremath{R}\xspace}}

\newcommand{\pj}{{\ensuremath{\bm j}\xspace}}

\newcommand{\pjM}{{\ensuremath{\pj^{\bm M}\xspace}}}
\newcommand{\pv}{{\ensuremath{\bm v}\xspace}}

\newcommand{\pvM}{{\ensuremath{\bm \pv^{\bm M}}\xspace}}

\newcommand{\cksumT}{{\ensuremath{R_\cksum}\xspace}}

\newcommand{\rowball}{{\ensuremath{{\cB_{d, \FF}}}\xspace}}
\newcommand{\kmid}{{\ensuremath{{k^{\texp}}}\xspace}}
\newcommand{\tmid}{{\ensuremath{{t^{\texp}}}\xspace}}
\newcommand{\cLmid}{{\ensuremath{{\cL_{\tmid}^{\kmid}}}\xspace}}

\newcommand{\ncol}{{\ensuremath{{L}}\xspace}}

\newcommand{\cQ}{{\ensuremath{{\mathcal{Q}}}\xspace}}

\newcommand{\cksumv}{{\ensuremath{\chi}\xspace}}

\renewcommand{\secpar}{{\ensuremath{{\sigma}}\xspace}}

\newcommand{\NP}{\textsf{NP}\xspace}

\newcommand{\LDP}{\textsf{LDP-IP}\xspace}
\newcommand{\IP}{\textsf{IP}\xspace}

\newcommand{\IPP}{\textsf{IPP}\xspace}

\newcommand{\TS}{{\ensuremath{\mathsf{TISP}\xspace}}}

\newcommand{\depth}{\ensuremath{{\mathsf{depth}}\xspace}}
\newcommand{\size}{\ensuremath{{\mathsf{size}}\xspace}}
\newcommand{\Vtime}{\ensuremath{{\mathsf{Vtime}}\xspace}}
\newcommand{\Ptime}{\ensuremath{{\mathsf{Ptime}}\xspace}}

\newcommand{\transcMat}{\ensuremath{{\bm a}\xspace}}

\newcommand{\I}{\ensuremath{{\mathcal{I}}\xspace}}
\newcommand{\Q}{\ensuremath{{\mathcal{Q}}\xspace}}
\newcommand{\C}{\ensuremath{{\Phi\xspace}}}

\renewcommand{\pp}{\ensuremath{\mathsf{pp}}}

\newcommand{\xmid}{\ensuremath{{\bm x^{\texp}}}}
\newcommand{\xred}{\ensuremath{{\bm x^{\tred}}}}
\newcommand{\kred}{\ensuremath{{k^{\tred}}}}
\newcommand{\Mred}{\ensuremath{{\bm M^{\tred}}}}

\newcommand{\Mx}{\ensuremath{{\bm M}}}
\newcommand{\CRR}{\ensuremath{\Psi}\xspace}

\newcommand{\Cmid}{\ensuremath{{\Psi_{\texp}}}\xspace}
\newcommand{\CCmid}{\ensuremath{{C_{\texp}}}\xspace}
\newcommand{\CIPP}{\ensuremath{{\Psi_{\IPP}}}\xspace}

\newcommand{\Cbase}{\ensuremath{\C_{\mathsf{base}}}\xspace}
\newcommand{\CCbase}{\ensuremath{C_{\mathsf{base}}}\xspace}
\newcommand{\Ccontract}{\ensuremath{\Psi_{\texp}}\xspace}
\newcommand{\CCcontract}{\ensuremath{C_{\texp}}\xspace}
\newcommand{\Cred}{\ensuremath{\Psi_{\tred}}\xspace}
\newcommand{\CCred}{\ensuremath{C_{\tred}}\xspace}
\newcommand{\CBatch}{\ensuremath{\Psi_{\Batch}}\xspace}
\newcommand{\CCBatch}{\ensuremath{C_{\Batch}}\xspace}
\newcommand{\CSmallBatch}{\ensuremath{\Psi_{\SmallBatch}}\xspace}
\newcommand{\CCSmallBatch}{\ensuremath{C_{\SmallBatch}}\xspace}
\newcommand{\Creduce}{{\ensuremath{\C_{\mathsf{reduce}}}\xspace}}
\newcommand{\CCreduce}{{\ensuremath{C_{\mathsf{reduce}}}\xspace}}

\newcommand{\protcplx}[1]{{\ensuremath{(\ell_{#1}, a_{#1}, \Ptime_{#1}, \Vtime_{#1})}}}
\newcommand{\BatchParam}{{\ensuremath{S, T, t, k, \lrag{\C}}}}

\newcommand{\Fbits}{{\ensuremath{\polylog{\abs{\FF}}}}}

\newcommand{\Flog}{{\ensuremath{\log{\abs{\FF}}}}}
\newcommand{\logt}{{\ensuremath{\tau}}}
\newcommand{\recnodes}{{\ensuremath{N(T,k)}}}
\newcommand{\secloc}{{\ensuremath{\secpar_{\mathsf{loc}}}}}

\newcommand{\Fboundstmt}{{\ensuremath{c_02^{4\secloc} T^{10}(k+1)^4 S^{c_1} \cdot (\depth(C) \log \size(C)+1)}}}

\newcommand{\FboundT}{{\ensuremath{2^{\secpar + 5}(\pvalT \log k \log L)^c}}}
\newcommand{\FboundIPP}{{\ensuremath{C_\GKR\cdot 2^{\secpar + 4c + 7} ((\secpar dL\log (\abs{\FF}(k+1)))^c + \depth(C) \log \size(C)+1)}}}
\newcommand{\aboundIPP}{{\ensuremath{L + \poly(d)}}}

\newcommand{\Tbound}{{\ensuremath{8dL (\log k + \Flog) + \secpar + 2}}}
\newcommand{\epstk}{{\ensuremath{\eps(t, k)}}}
\newcommand{\epsIPPconcrete}{{\ensuremath{2^{-\secpar-2}}}}
\newcommand{\epsmid}{{\ensuremath{\eps(\tmid, \kmid)}}}
\newcommand{\epsred}{{\ensuremath{\eps(t, \kred)}}}
\newcommand{\GKRbound}{{\ensuremath{2^{-\secloc}}}}
\newcommand{\IPPbound}{{\ensuremath{2^{-\secloc}}}}

\newcommand{\baset}{{\ensuremath{t_{\mathsf{base}}}}}
\newcommand{\basep}{{\ensuremath{p_{\mathsf{base}}}}}
\newcommand{\basek}{{\ensuremath{k_{\mathsf{base}}}}}
\newcommand{\baseq}{{\ensuremath{q_{\mathsf{base}}}}}

\newcommand{\dbound}{{\ensuremath{96\secloc\lambda \log (Tk)}}}
\newcommand{\dboundm}{{\ensuremath{16\secpar \log k}}}
\newcommand{\dboundIPP}{{\ensuremath{48\secpar\log k}}}

\newcommand{\Reduce}{{\textsf{{RR}}\xspace}}

\FOCSfalse    
\title{Towards a Doubly Efficient $\mathsf{IP} = \mathsf{PSPACE}$}
\usepackage{authblk}

\author[1]{Liyan Chen}
\author[1]{Matthew M. Hong}
\author[1]{Yael Tauman Kalai}
\author[1]{Zoe Xi}

\affil[1]{Massachusetts Institute of Technology}
\date{\today}

\begin{document}
\maketitle

\begin{abstract}
We show that every language in $\mathsf{PSPACE}$ decidable by a Turing machine in time
$T(n)=n^{O(\log n)}$
admits a doubly efficient interactive proof system: the prover runs in time $\mathsf{poly}(T(n))$, and the verifier runs in time $\mathsf{poly}(n)$.
This extends the best previously known regime for such proof systems from
$T(n)=n^{O\!\left(\sqrt{\log n/\log\log n}\right)}$,
established by Berger, Goyal, Hong, and Kalai (FOCS 2025), to $T(n)=n^{O(\log n)}$.

Beyond improving the range of $T$, our protocol is substantially simpler than previous doubly efficient proofs for time-bounded $\mathsf{PSPACE}$.
Earlier constructions proceed indirectly: they first build batch interactive proofs and then invoke them as a black box to obtain doubly efficient protocols. 
In contrast, we give a direct construction.
This not only simplifies the proof but also points to a more promising route for future improvements.
\end{abstract}

\thispagestyle{empty}
\newpage
\enlargethispage{1cm}
\tableofcontents
\addtocontents{toc}{\protect\thispagestyle{empty}}
\thispagestyle{empty}
\newpage
\pagenumbering{arabic}
\section{Introduction}

Verification is one of the most fundamental concepts in computer science, and motivated the study of the  complexity classes $\NP$ \cite{cook1971complexity}, ${\sf IP}$~\cite{GolMicRac89,BabaiM88} and ${\sf MIP}$~\cite{STOC:BGKW88,FOCS:BabForLun90}, among others. This study has led to foundational concepts in cryptography such as zero-knowledge proofs \cite{GolMicRac89}, which are used as building blocks in many cryptographic protocols used today. It has also inspired fundamental notions such as  probabilistically checkable proofs (PCPs) \cite{STOC:BFLS91,FOCS:FGLSS91,FOCS:AroSaf92,FOCS:ALMSS92},
interactive PCPs~\cite{ICALP:KalRaz08},
and interactive oracle proofs~\cite{TCC:BenChiSpo16,STOC:ReiRotRot16}, which in turn have led to breakthrough results in hardness of approximation and to the construction of ${\sf SNARG}$s (succinct non-interactive arguments), which are used in many blockchain applications.

\paragraph{Interactive proofs ($\IP$s)} 

The power of interactive proofs was demonstrated  by the celebrated 
$\IP=\mathsf{PSPACE}$ theorem \cite{JACM:LFKN92,Shamir92}.
Specifically, it was proven that the correctness of any time-$T$, space-$S$ computation can be verified by a $\poly(S,n)$-time verifier, via an interactive proof. However, this comes at a price:  The time required by the prover to convince the verifier of the correctness is $2^{\Omega(S\cdot \log S)}$. While in those works the runtime of the prover was not a parameter of interest, and the prover was thought of as being all-powerful (and was even named after the famous wizard Merlin~\cite{BabaiM88}), this blowup in the prover's runtime makes these proof systems completely impractical.

\paragraph{Doubly efficient $\IP$s}  
The work of \cite{JACM:GolKalRot15} initiated the study of \emph{doubly efficient} $\IP$s, in which the prover's runtime is required to be at most polynomial in $T$ (the time it takes to run the underlying computation), and the verifier is required to run in time significantly less than~$T$ (otherwise, such interactive proofs are trivial).  
They constructed a doubly efficient $\IP$ for any computation that can be performed by a (log-space uniform) circuit of depth $D$
and size $T$, where the verifier's runtime is 
$D\cdot\polylog(T)+ \tilde{O}(n)$, and the communication complexity is $D\cdot\polylog(T)$, where $n$ is the input length. 
In particular, this implies an improved 
$\IP=\mathsf{PSPACE}$ theorem where the verifier runs in time
$\poly(S)+ 
\tilde{O}(n)$, the communication complexity is $\poly(S)$, and the prover runs in time
$2^{O(S)}$. 
This result yields a doubly efficient $\IP$ for log-space computations.  

The next significant advancement was due to Reingold, Rothblum, and Rothblum \cite{STOC:ReiRotRot16}, who constructed a doubly efficient $\IP$, where the prover's runtime is $\poly(T)$ and the verifier's runtime is $\poly(n)$, for every language computable in polynomial space and time $T=n^{O((\log n)^\delta)}$ for a sufficiently small constant $\delta>0$.
Very recently, \cite{FOCS:BGHK25} improved this result by constructing a doubly efficient $\IP$ for every language in polynomial space and time $T=n^{O\!\left(\sqrt{\log n/\log\log n}\right)}$.

The main technical contribution in both of these works is an efficient way to batch unambiguous interactive proofs.  An unambiguous interactive proof is an interactive proof with the guarantee that if the prover ever deviates from the unique prescribed strategy, then it will be rejected with high probability. Both results mentioned above \cite{STOC:ReiRotRot16,FOCS:BGHK25} show that if membership in a language $\cL$ can be proven via a doubly efficient unambiguous $\IP$, then one can prove that $x_1,\ldots,x_k\in\cL$ via a doubly efficient unambiguous $\IP$  where the verifier runs in time significantly less than running these $k$ proofs.
These works then use this batch unambiguous $\IP$ in a black-box way to construct a doubly efficient $\IP$.
This latter part is quite straightforward. 
Indeed, the improvement in \cite{FOCS:BGHK25} is in constructing a more efficient batch unambiguous $\IP$, which immediately yields a more efficient doubly efficient $\IP$.

\subsection{Our Result}

We construct a doubly efficient $\IP$ for every language that is computable in polynomial space and time $T=n^{O(\log n)}$, thus improving both previous results.

\begin{theorem}[Our Doubly Efficient $\IP$, Informal]
    \label{thm:informal-deip}
    There exists a doubly efficient $\IP$ (\ie an interactive proof with a $\poly(T)$ prover and a $\poly(n)$ verifier) for every language that is computable in polynomial space and time $T=n^{O(\log n)}$.
\end{theorem}

More generally, we prove the following theorem.

\begin{theorem}[Our General Doubly-Efficient $\IP$, Informal]
    \label{thm:informal-deip-gen}
    There exists an $\IP$ for every language that is computable in time $T$ and space $S$, where the prover runs in time $\poly(T)$ and the verifier runs in time $\poly(S, 2^{\sqrt{\log T}})$.
\end{theorem}

Our construction is significantly simpler than the constructions in both of these prior works.  In particular, the construction is direct.  As opposed to previous works, which use a batch unambiguous $\IP$ as a black box, we directly batch ${\sf DTISP}$ computations. We believe that this direct construction has the potential to lead to further improvements, while constructions that use batch (unambiguous) $\IP$ as a black box are subject to barriers. Specifically, to get a doubly efficient $\IP$ for languages computable in time $T = n^{\omega(\log n)}$ and polynomial space using a batch (unambiguous) $\IP$ as a black box, one would need a ``rate-1'' batch (unambiguous) $\IP$, i.e., one where the communication complexity for  proving $k$ statements is $(1+o(1))\cdot{\sf cc}$, where  ${\sf cc}$ is the communication complexity of proving a single statement. 
Constructing such a rate-1 batch (unambiguous) $\IP$ seems challenging.
On the other hand, we believe that our protocol points to a more promising route for future improvements. 

\subsection{Additional Related Work}

Our protocol builds on two sub-protocols from prior works.  The first is the doubly efficient $\IP$ for bounded depth circuits from \cite{JACM:GolKalRot15} which we mentioned above. 
The second is an ``instance reduction'' protocol from \cite{STOC:ReiRotRot16,RRR18,TCC:RotRot20}, which roughly says that if $x_1,\ldots,x_k\in\cL$ is far from being true (i.e., one needs to change, say $d$, of the statements $x_i$ to make it true), then there is a so-called ``instance reduction'' protocol that reduces checking that $x_1,\ldots,x_k\in \cL$ to checking that $\tilde{O}(k/d)$ of these instances are in $\cL$ (with an additional linear test on these $\tilde{O}(k/d)$ instances). Such an instance-reduction protocol is special case of an interactive proof of proximity ($\IPP$) \cite{STOC:RotVadWig13}, which is an efficient protocol that, given an input $x$ far from satisfying some predicate $P$, i.e., one needs to change at least $d$ coordinates in $x$ to satisfy $P$, reduces
verifying that  $P(x) = 1$ to  verifying that $P'(x_S) = 1$, where $P'$ is a related  predicate and $S$ is a subset of
indices of size $\tilde{O}(|x|/d)$.

We also mention that there is a large body of work, starting with \cite{BrassardChaumCrepeau88,FOCS:Micali94}, on constructing computationally sound proofs, where soundness is guaranteed to
hold only against computationally bounded cheating provers. We do not elaborate on these works
here, since they are less relevant for our work.

\section{Technical Overview}
\label{sec:overview}

\newcommand{\cf}[0]{\mathsf{cf}}
\newcommand{\Set}[1]{\left\{ #1 \right\}}

In this section, we give an overview of our approach for  proving membership in a  $\TS(T, S)$ language, 
where the verifier is given a Turing machine $M$ and two configurations $\cf_0$ and  $\cf_T$ of size bounded by $S$, and it wishes to verify whether running $M$ for $T$ steps from $\cf_0$ reaches $\cf_T$.

One natural idea is for the prover to send intermediate configurations: let $\lambda$ be a fixed parameter,\footnote{We will later discuss how to set this parameter. } and the prover sends 
\[\cf_{T/\lambda}, \cf_{2T/\lambda}, \ldots, \cf_{(\lambda-1)T/\lambda}\] 
as the $\lambda$ intermediate configurations after $\frac{T}{\lambda}, \frac{2T}{\lambda},\ldots, \frac{(\lambda-1)T}{\lambda}$ steps, respectively. This reduces the task of verifying a $T$-time statement ($\cf_0 \overset{T}{\to} \cf_T$) to the task of verifying $\lambda$ many $T/\lambda$-time statements \[\cf_{i\cdot T/\lambda} \overset{T}{\to} \cf_{(i+1)\cdot T/\lambda}~\mbox{ for all }~0 \leq i < \lambda.
\]

Prior works, starting with \cite{STOC:ReiRotRot16}, and continuing with \cite{RRR18,TCC:RotRot20,FOCS:BGHK25}, showed that verifying $k$ statements $x_1, \ldots, x_k$ is easier than verifying each of them separately.  Roughly speaking, this is based on the following initial observation: if we are guaranteed that a constant fraction of $x_1, \ldots, x_k$ are false, then it suffices to randomly sample a small subset of $x_1, \ldots, x_k$ and perform verification on that subset. Following this observation,  previous works on batching (unambiguous) $\NP$ statements \cite{RRR18,TCC:RotRot20} and batching (unambiguous) interactive proofs \cite{STOC:ReiRotRot16,  FOCS:BGHK25} manage to cleverly reduce the number of instances by a constant factor, even when initially only a single claim is false, by running sub-protocols that add sufficiently many constraints on the corresponding  witnesses $(w_1,\ldots,w_k)$, or on the transcripts in the case of (unambiguous) $\IP$, such that with these additional constraints many of the statements become false.  They then repeat this recipe until they are left with a constant number of instances, which the verifier can verify on its own or via an interactive proof. 

In order to obtain a doubly efficient interactive proof for $\TS(T, S)$, previous works, as well as this work, first break down the statement into $\lambda$ statements, each in $\TS(T/\lambda, S)$. They then apply a batching protocol for these $\lambda$ statements, reducing the task to verifying a constant number of statements in $\TS(T/\lambda, S)$. These works then break down each of the remaining statements in $\TS(T/\lambda, S)$ into $\lambda$ statements, each in $\TS(T/\lambda^2, S)$, and apply the batching protocol again. This recipe is repeated until $T=O(1)$ (or until it is polynomial), in which case  the verifier can verify the statements on its own. 


What differentiates this work from previous works is the batching protocol used.  Previous works used a batch protocol for unambiguous {\em interactive proofs}, which is conceptually complex and brings a lot of technical difficulties. This work, on the other hand, directly  batches $\TS(T, S)$ statements. Our batching protocol recursively reduces both $k$ and $T$, and is surprisingly simple.

\paragraph{Core Ingredient: Interactive Proof of Proximity}

Before describing our protocol, we first introduce an important ingredient: an Interactive Proof of Proximity ($\IPP$). 

An $\IPP$ is a proof system that combines interactive proofs with property testing: a sublinear-time verifier, who can only query a few bits of a huge input $x$, interacts with a prover to decide whether $x$ satisfies some property $\Pi$ or is far (in Hamming distance) from every string that satisfies property $\Pi$. 
Completeness requires that if $x$ satisfies property $\Pi$, then the honest prover convinces the verifier to accept with probability~$1$. Soundness requires that if $x$ is far from satisfying $\Pi$ in Hamming distance, no prover strategy will make the verifier accept with more than a  small probability. 
Previous work \cite{TCC:RotRot20} constructs an $\IPP$ for any low-depth property with communication complexity $\tilde{O}(d)$ and query complexity $\tilde{O}(n/d)$, where $d$ is the distance threshold corresponding to the soundness guarantee. Furthermore, the $\IPP$ verifier makes \emph{non-adaptive} queries: it never queries the instance $x$ until the end of the interaction, at which point it samples a random set of query indices $Q$ (depending on its randomness in previous rounds and this end phase), and the final decision is a low-depth predicate that depends only on $x_Q$. 

This work uses a generalization of $\IPP$ from \cite{FOCS:BGHK25},  called row-$\IPP$, which treats the huge input $x$ as a large matrix $M$ with many rows. The verifier can only query a few rows in $M$, and soundness holds for any $M$ with a large row-distance from the property $\Pi$, i.e., for any $M$ for which every $M'$ that satisfies $\Pi$ differs from $M$ in many rows. Row-$\IPP$ is known for parameters similar to those of standard $\IPP$; specifically, for $n \times m$ matrices and distance parameter~$d$, \cite{FOCS:BGHK25} adapts the $\IPP$ from \cite{TCC:RotRot20} into a row-$\IPP$ for any low-depth property with communication complexity $\tilde{O}(d m)$ and query complexity $\tilde{O}(n/d)$. The row-$\IPP$ verifier also makes \emph{non-adaptive} queries: the final decision is a low-depth predicate that depends on a few rows in $M$. 

\paragraph{Batching TISP: A 4-Step Construction.}

Our protocol $\mathsf{Batch}(T, k)$ for proving $k$ $\TS(T, S)$ statements is as follows.
Assume that we are proving $k$ statements $(x_1, \ldots, x_k)$, where 
\[x_i : \cf_0^i \overset{T}{\to} \cf^i_{T}
\]
is a time-$T$ computation. An immediate idea is to apply the famous GKR protocol \cite{JACM:GolKalRot15} to verify the correctness of all $k$ $\TS(T, S)$ statements. However, the verification time of the GKR protocol depends on the \emph{depth} of the computation, which is $O(T \poly(S) \log k)$ for the union of $k$ $\TS(T, S)$ statements: The main overhead for applying GKR is $T$. Our main motivation for the following steps is to reduce $T$ of the $\TS$ statements we need to verify. 

We let $w_i$ be the $\lambda$ intermediate configurations, i.e.,
\[
w_i=(\cf^i_{T/\lambda}, \cf^i_{2T/\lambda}, \ldots, \cf^i_{(\lambda-1)T/\lambda}).
\] 

Consider the following $k$-by-$\lambda$ matrix, denoted by $M$, where the $i$-th row consists of $w_i$, and to simplify exposition, we also include in the $i$-th row the statement $x_i = (\cf^i_0, \cf^i_{T})$ in the first and last columns, respectively:

\[
M = \begin{pmatrix}
\cf^1_0      & \cf^1_{T/\lambda}      & \cf^1_{2T/\lambda}      & \cdots & \cf^1_{(\lambda-1)T/\lambda}      & \cf^1_{T} \\
\cf^2_0      & \cf^2_{T/\lambda}      & \cf^2_{2T/\lambda}      & \cdots & \cf^2_{(\lambda-1)T/\lambda}      & \cf^2_{T} \\
\vdots       & \vdots                    & \vdots                     & \ddots & \vdots                               & \vdots            \\
\cf^k_0      & \cf^k_{T/\lambda}      & \cf^k_{2T/\lambda}      & \cdots & \cf^k_{(\lambda-1)T/\lambda}      & \cf^k_{T}
\end{pmatrix}.
\]

The matrix $M$ defines $k\lambda$ $\TS(T/\lambda, S)$ statements, by taking each pair of consecutive elements in a row as the start and end configurations. We denote the $k\lambda$ statements defined by $M$ as ${\bm x}_M$.

Let $\lambda, d$ be two global parameters. The protocol $\mathsf{Batch}(T, k)$ is as follows:

\begin{enumerate}
    \item \textbf{Sending Checksums.} The prover computes checksums of~$M$, denoted by $\cksum$, such that any two matrices that are consistent with $\cksum$ have row-distance at least~$(2d+1)$ (i.e., they differ in at least~$(2d+1)$ rows).
    It sends $\cksum$ to the verifier.
 \item \textbf{Reducing $T$.} The prover and verifier engage in $\mathsf{Batch}(T / \lambda, k \lambda)$ over ${\bm x}_M$. However, we delay the verification of $\mathsf{Batch}(T / \lambda, k \lambda)$: the verifier does not access ${\bm x}_M$ directly; instead, it outputs a low-depth predicate $\psi$ about ${\bm x}_M$ such that $\psi({\bm x}_M) =0$ if ${\bm x}_M$ is not correct. Define $\Psi$ as $\Psi(M) = \psi({\bm x}_M)$.

    \begin{remark}
      The purpose of these two steps is to create distance:  Either every $M'$ that is consistent with $\cksum$ is $d$-far from $M$, or there exists an $M'$ consistent with $\cksum$ that is $d$-close to $M$ for which we will show $\Psi(M')=0$.
    \end{remark}
   
    \item \textbf{IPP.} The prover and verifier engage in a row-$\IPP$ proving that $M$ is close in row-distance to satisfying the following property (checkable by a low-depth circuit):
    \begin{itemize}
        \item $M$ is consistent with $\cksum$;
        \item $\Psi(M) = 1$.
    \end{itemize}
    The distance parameter of this row-IPP is $d$, and the verifier at the end needs to query $k / d$ rows $M[\cS, :]$ in $M$. Note that $M[\cS, :]$ are the $\lambda$ intermediate configurations of $\Set{x_i: i \in \cS}$. The output of the IPP verifier is a low-depth predicate $\C$ about $M[\cS, :]$. 
    \item \textbf{Reducing $k$.} The prover and verifier engage in $\mathsf{Batch}(T, k/d)$ over $k/d$ $\TS(T, S)$ statements $\Set{x_i: i \in \cS}$, with the goal of checking that $M[\cS, :]$, defined by the $k/d$ statements $\Set{x_i: i \in \cS}$, satisfies $\C$, and that the $k/d$ statements $\Set{x_i: i \in \cS}$ are valid. We recurse over a batched protocol $\mathsf{Batch}(T, k/d, \C)$ that checks a low-depth predicate $\C$ on the \emph{middle configurations} and follows exactly the same recursive construction that we are describing. 
    In this technical overview, we forget $\C$ for simplicity, and only say we want to check if all of $\Set{x_i: i \in \cS}$ are valid.
\end{enumerate}

The base case is when $T = O(1)$ or $k = O(1)$. 
\begin{itemize}
    \item \textbf{$T$ is small}. In this case, all $k$ $\TS(T, S)$ statements can be verified by a low-depth circuit. Thus, the prover and verifier can engage in the GKR protocol, which returns a low-degree extension check on the statements.
    \item \textbf{$k$ is small}. In this case, the prover sends the entire $M$ in the clear, and the prover and verifier engage in $\mathsf{Batch}(T/\lambda, k\lambda)$. 
\end{itemize}

The above is the full description of our construction. We analyze it and provide concrete parameters below. 

\paragraph{Protocol Analysis: Win-Win Argument}

In what follows, we say that the   ``unique'' witness  (or true configurations) corresponding to the $i$-th statement $x_i=(\cf^i_0,\cf^i_{T})$ consists of the middle configurations obtained by computing honestly from $\cf^i_0$ (even if the statement $x_i$ is false). Thus, $M$ is well-defined even if some statements are false.
Our analysis goes through the four steps in the protocol.

\paragraph{Step 1: Sending Checksums.} In the first step of the protocol, the prover sends checksums, denoted by $\cksum$, for all the columns of~$M$. The size of each column of $\cksum$ is set to be $\tilde{O}(d)$ so that one can (uniquely) decode any deviation in at most $d$ rows.  It is convenient to think of $d=\lambda$ since we will indeed set these two parameters to be equal. The total size of $\cksum$ is $\tilde{O}(d) \cdot |w_i| = \tilde{O}(d\cdot \lambda\cdot S)$, where $S$ is the size of each configuration. 

We distinguish between two cases, both of which take us a step closer to catching the cheating prover:

\begin{itemize}
    \item {\bf Case 1:} One needs to change more than $d$ rows of the unique witnesses $M$ to be consistent with $\cksum$.  In this case,
    $M$ already has a large row-distance to the property proven by row-$\IPP$ in step 3. This is an easy   case as we can use the soundness guarantee of the row-$\IPP$. In this case step 2 of the protocol is not necessary since we already have the distance needed for the IPP soundness. 
    \item {\bf Case 2:}  One needs to change fewer than $d$ rows of the unique witnesses $M$ to be consistent with $\cksum$. In this case, we cannot directly say that the distance condition for row-$\IPP$ soundness is satisfied.
    However, this case still offers us a win:  $\cksum$ uniquely determines the witnesses sent by the prover, since for each column, $\cksum$ uniquely decodes up to $d$ deviations.  In other words, if we haven't created distance then the prover has committed to a matrix $M'$ that is $d$-close to $M$ through $\cksum$.
\end{itemize}

 Note that the  committed  matrix $M'$ is a fixed $k$-by-$\lambda$ matrix that contains $k \lambda$ configurations, corresponding to $k \lambda$ statements ${\bm x}_{M'}$, where each statement corresponds to  a $T/\lambda$-time computation. We can assume that $M'$ and $M$ have the same  starting and ending configurations in each row, since the verifier can efficiently check this. 
This implies that one of these $k \lambda$ statements in ${\bm x}_{M'}$ must be false. 

\paragraph{Step 2: Reducing $T$.} The purpose of this step is to create the distance needed for the IPP if we are in Case 2 in the last step.
At step 2, the prover and verifier run the protocol $\mathsf{Batch}(T/\lambda, k \lambda)$. Recall that the verifier does not access the instance, but only outputs a predicate $\psi$ about this instance. In our soundness analysis, we regard this $\mathsf{Batch}(T/\lambda, k \lambda)$ as executed w.r.t. ${\bm x}_{M'}$, the $k\lambda$ statements defined by $M'$. Therefore, by soundness of $\mathsf{Batch}(T/\lambda, k \lambda)$, with high probability $\psi({\bm x}_{M'}) = 0$, i.e. $\Psi(M') = 0$.

\paragraph{Step 3: IPP.} At step 3, the prover and verifier engage in a row-$\IPP$ proving that $M$ is close in row-distance to the property defined by the $\cksum$ from step 1 and $\Psi$ from step 2. Recall that in step 1 we distinguished between two cases: $M$ is close (in row distance) to $\cksum$ or not. We next claim that in both cases $M$ is far (in row distance) from the property that row-$\IPP$ is proving:

\begin{itemize}
    \item \textbf{Case 1:} $M$ is already far from the satisfying checksums.
    \item \textbf{Case 2:} The only possible matrix close to $M$ and satisfying the checksums is $M'$. However, $\Psi(M') = 0$. 
\end{itemize}

Therefore, by soundness of row-IPP, with high probability, the verifier obtains a set $\cS$ and predicate $\C$ such that $\C(M[\cS, :]) = 0$.

\paragraph{Step 4: Reducing $k$.} We are left with checking $\C(M[\cS, :])=1$. If we don't need to check $\C$ but need only check whether $\Set{x_i: i \in \cS}$ are correct, then simply running $\mathsf{Batch}(T, k/d)$ is sufficient. In our actual construction, we instead check if $\C$ is satisfied by the intermediate configurations $M[\cS, :]$ uniquely determined by $\Set{x_i : i \in \cS}$. This turns out to be extremely close to checking correctness of $\Set{x_i: i \in \cS}$, and we recursively build $\mathsf{Batch}(T, k, \C)$ for the above condition with the same 4-step construction.

\paragraph{Efficiency analysis}

In the construction above, the $\mathsf{Batch}(T, k)$ protocol runs a $\mathsf{Batch}(T/\lambda, k\lambda)$ protocol and a $(T, k/d)$ protocol with additional $\tilde{O}(\lambda\cdot d\cdot  S)$ communication and $\poly(\lambda, d, S)$ verification time. Taking $\lambda = d$, the verification time satisfies the following transition (below we use notation $O_{\lambda, S}$ to hide $\poly(\lambda, S)$ terms): 
\begin{itemize}
    \item $\Vtime_{1, k} = O_{\lambda, S}(1)$;
    \item $\Vtime_{T, 1} = \Vtime_{T/\lambda, \lambda} + O_{\lambda, S}(1)$;
    \item $\Vtime_{T, k} = \Vtime_{T/\lambda, k\lambda} + \Vtime_{T, k/\lambda} + O_{\lambda, S}(1)$.
\end{itemize}

Standard calculation shows that $\Vtime_{T, 1}$ equals $O_{\lambda, S}(C_{\log_{\lambda} T})$, where $C_{\ell} \approx 2^{2\ell}/\poly(\ell)$ is the $\ell$-th Catalan Number. Taking $\lambda = 2^{\sqrt{\log T}}$ gives us the desired $\Vtime_{T, 1} = O_{\lambda, S}(2^{\sqrt{\log T}})$, and in particular $\Vtime_{T, 1} = \poly(n)$ when $T = n^{O(\log n)}$.


\section{Preliminaries}
We adhere to the conventions in \cite{FOCS:BGHK25},
hence many definitions and lemmas in this section are taken verbatim from that work.
\begin{itemize}
    \item Lower case letters $a,b$  mean scalars,
          while bolded lower case letters $\bm a, \bm b$ mean vectors.
          Upper case bold letters like $\bm A, \bm B$ are matrices.
    \item For an integer $n$, we denote by $[n]$ the set $\{1,2,\ldots,n\}$.
          When it is clear from context, we use $[0, n]$ to denote the set $\{0,1,\ldots,n\}$.
          We let $\FF$ denote a finite field.
          $\GF(2)$ denotes the finite field with 2 elements.
    \item
          For a vector $\bm u \in \FF^n$,
          and a subset $\cS \subset [n]$,
          $\bm u\vert_\cS \in \FF^{\abs{\cS}}$ denotes the subvector of $\bm u$ indexed by $\cS$.
          For a sequence of vectors $\set{\bm u_i}_{i \in \I}$, where each $\bm u_i\in \FF^n$,
          the notation $(\bm u_i)_{i \in \I} \in \FF^{n \times |\I|}$ represents the matrix whose columns are the vectors $\bm u_i$ for every $i \in \I$.
    \item
          Given a matrix $\bm A = (\bm a_1,\ldots,\bm a_n) \in \FF^{m \times n}$,
          $\bm a_j \in \FF^m$ denotes its $j$-th column.
          We also use $\bm A[i, :]$ to denote the $i$-th row of $A$.  
          We use $\bm A[:, 0:j]$ or $\bm A[:, :j]$
          to denote the submatrix of $\bm A$ consisting of the first $j$ columns, and use $\bm A[:, -j:]$ to denote the submatrix of $\bm A$ consisting of the last $j$ columns.
           For a subset $\cS \subset [m]$,
          $\bm A[\cS,:]$ denotes the submatrix of $\bm A$ consisting of rows indexed by $\cS$. 
    \item
          $\Delta(\bm u,\bm v)$ is the (absolute) Hamming distance between the vectors $\bm u$ and $\bm v$.
          Denote by $\Delta(\bm u)$ the vector $\bm u$'s \emph{Hamming weight},
          defined to be the number of non-zero elements in $\bm u$.
          For any $d \in \NN$,
          two vectors $\bm u, \bm v$ are \emph{$d$-close} (in Hamming distance) if $\Delta(\bm u,\bm v) \le d$, and \emph{$d$-far} otherwise.
          On a linear space $\cU\subset \FF^a$,
          let $\Delta(\cU) \coloneqq \min_{\bm u \in \cU, \bm u \neq \bm 0}\Delta(\bm u)$.
    \item
          Given a Boolean circuit $\cV$,
          $\size(\cV)$ is the number of gates in the circuit,
          and $\depth(\cV)$ is the depth of the circuit.
    \item
          Given a Turing machine $\cM$,
          we use $\lrag{\cM}$ to denote its constant-size description.
          Denote by $\TS(T, S)$ the class of languages decidable by a Turing machine in time $T(n)$ and space $S(n)$.
\end{itemize}

\subsection{Finite Fields and Distances on Matrices}
All finite fields $\FF$ considered in this work are always \emph{constructible} in the following sense.
\begin{definition}[Constructible Field Ensemble]
    We say a field ensemble $\FF = (\FF_n)_{n \in \NN}$ is \emph{constructible} if every element in $\FF_n$ has an
    $O(\log{\abs{\FF_n}})$-bit representation,
    and addition, multiplication, and inverses can be computed in $\polylog(\abs{\FF_n})$ time given the representations.
\end{definition}
It is well known that for every $S = S(n)$,
constructible field ensembles $\FF = (\FF_n)$ with $\abs{\FF_n} = \Theta(S)$ exist.
Moreover,
we make the (mild) assumption that addition, multiplication and inverses can be computed in $\tO(\log \abs{\FF_n})$ time,
where $\tO$ omits $\poly\log\log(\abs{\FF_n})$ factors.
This implies that degree-$d$ polynomials in $\FF[X]$ can be evaluated in $\tO(d\log\abs{\FF_n})$ time (with $\tO$ omitting $\polylog(d, \log \abs{\FF_n})$ factors).
These properties are satisfied in fields that support FFT. (See table $8.6$ in \cite{vzGG13})

An important metric defined on matrices, denoted by \emph{$\Delta_c$}, is as follows. 
\begin{definition}[$\Delta_c$-distance]
    \label{def:Deltac-dist}
    Let $\FF$ be a finite field.
    For matrices $\bm A = (\transcMat_1,\ldots, \transcMat_\ncol) \in \FF^{k \times \ncol}$ and $\bm B = (\bm b_1, \ldots, \bm b_\ncol)\in \FF^{k \times \ncol}$,
    the \emph{$\Delta_c$-distance} between $\bm A$ and $\bm B$, denoted by  $\Delta_c(\bm A, \bm B)$, is the maximum Hamming distance between corresponding columns of $\bm A$ and $\bm B$,
    i.e.,
    \[\Delta_c(\bm A, \bm B) = \max_{i \in [\ncol]}\Delta(\transcMat_i, \bm b_i).
    \]
    If $\Delta_c(\bm A, \bm B) \le d$,
    we say $\bm A$ and $\bm B$ are \emph{$\Delta_c$-$d$-close}. If $\bm A$ and $\bm B$ are not $\Delta_c$-$d$-close, then we say they are \emph{$\Delta_c$-$d$-far}.

    Furthermore, given a distance parameter $d \in \NN$ and $\bm A \in \FF^{k \times \ncol}$,
    define \emph{$\Delta_c$-$d$-ball} as
    \[\rowball(\bm A) \coloneqq \set{\bm A' \in \FF^{k \times \ncol}: \Delta_c(\bm A, \bm A') \le d}.
    \]
\end{definition}

\begin{remark}
    Observe that $\Delta_c$ is a lower bound on how many rows have to be modified to transform one matrix into another.
\end{remark}

\subsection{Succinct Descriptions of Sets and Functions}
\label{sec:desc}
\begin{definition}[Uniform Arithmetic Circuits]
    Let $\FF$ be a field.
    Let $C = \set{C_n : \FF^n \to \FF^m}_{n \in \NN}$ be a family of arithmetic circuits,
    consisting of fan-in 2 \ADD and \MULT gates over a field $\FF$.
    For any $f = f(n)$, we say that $C$ is $f$-space uniform if there exists a fixed $O(f(n))$-space Turing machine $\cM$ that, on input $1^n$, outputs the full description of the circuit $C_n$.
    When $n$ is clear from the context, we omit the subscript $n$ and write $C$ instead of $C_n$.
\end{definition}
An important special case is when $f(n) = \log(n)$,
in which case we say that $C$ is log-space uniform.
\begin{remark}
    Any $f$-space uniform Boolean circuit $C$ can be trivially extended to a $f$-space uniform arithmetic circuit $C'$ of the same size and depth over any field $\FF$.
\end{remark}
The following \emph{succinct descriptions} of sets can be used to recover the entire set.
\begin{definition}[Descriptions of Sets]
    A bit string $\lrag \cS \in \bin^B$ is a description of a set $\cS = \set{s_1, \ldots, s_k} \subset \bin^p$
    if there exists a (multi-output) $p$-space uniform
    circuit $G: [k]\times\bin^{B} \to \bin^{p}$ of fan-in 2,
    called its \emph{implementation circuit},
    such that $G(i, \lrag \cS) = s_i$ for all $i \in [k]$.
    The description is succinct if $\abs{\lrag \cS} = B < k \cdot p$.
    \label{def:setdesc}
\end{definition}

Similarly,
\emph{succinct descriptions} of functions can be used to configure a uniform circuit family to implement the function.
\begin{definition}[Descriptions of Functions]
    Let $\FF$ be a field.
    We say that $\lrag \C \in \bin^B$ is a description of a function $\C :\FF^n \to \FF^m$ if there exists a log-space uniform circuit $C:\FF^{n + B} \to \FF^m$ of fan-in 2,
    called its \emph{implementation circuit},
    such that $C(x, \lrag{\C}) = \C(x)$ for all $x \in \FF^n$.
    The description is succinct if $\abs{\lrag \C} = B < \size(\C)$.
    \label{def:circdesc}
\end{definition}
When we write $\lrag \C = \top$,
we mean that $\C$ is the trivial predicate that outputs 1 on all inputs (i.e., it imposes no constraint on $x$).

The Turing machines that generate the uniform $G$ and $C$ in \Cref{def:setdesc,def:circdesc} take in their ``shape parameters'' $(k, 1^B, 1^p)$ and $(1^n, 1^B)$ as input, respectively.

\subsection{Unique-Decoding Checksums}\label{sec:prelim:cksum}
Let $\rho \in \NN$ be a \emph{deviation radius} and $\cksumT \in \NN$ be a checksum-length parameter.
We use syndromes of linear error-correcting codes as checksums.
\begin{definition}[Unique-Decoding Checksums]
    \label{def:unique-decoding-checksum}
    Let $k,\rho,\cksumT \in \NN$ and let $\FF$ be a field.
    A function $\cksum_\rho:\FF^k \to \FF^\cksumT$ is a \emph{$\rho$-unique decoding checksum function}
    if for any $\bm m \in \FF^k$,
    and for any $\bm m', \bm m'' \in \FF^k$ that are both $\rho$-close to $\bm m$ and $\bm m' \neq \bm m''$,
    $\cksum_\rho(\bm m') \neq \cksum_\rho(\bm m'')$.
\end{definition}
The term \emph{unique decoding} refers to the fact that if we know $\bm m'$ is $\rho$-close to some (fixed) $\bm m$,
then we can uniquely determine $\bm m'$ given $\cksum_\rho(\bm m')$.

\begin{proposition}[Linear-code checksums]
    Let $\mathcal{C}\subseteq\FF^k$ be a linear code of minimum distance $D$ and codimension $\cksumT$,
    and let $H\in\FF^{\cksumT\times k}$ be a parity-check matrix for $\mathcal{C}$.
    For any $\rho\in\NN$ satisfying $2\rho < D$,
    the syndrome map
    \[
        \cksum_\rho(\bm m) \coloneqq H\bm m
    \]
    is a $\rho$-unique decoding checksum function.
    Moreover, with $\tO$ hiding polylogarithmic factors, $\cksum_\rho$ can be evaluated in $\tO(k\cksumT(\Fbits))$,
    and its implementation circuit over $\FF$ has size $O(k\cksumT)$ and depth $\tO(1)$.
\end{proposition}
\begin{proof}
    Fix $\bm m\in\FF^k$ and let $\bm m',\bm m''\in\FF^k$ be distinct vectors that are both $\rho$-close to $\bm m$.
    Then $\bm z=\bm m'-\bm m''$ is nonzero and has Hamming weight at most $2\rho$.
    If $\cksum_\rho(\bm m')=\cksum_\rho(\bm m'')$,
    then $H\bm z=H\bm m'-H\bm m''=\bm 0$ by linearity,
    so $\bm z\in\ker(H)=\mathcal{C}$.
    This is a nonzero codeword of Hamming weight at most $2\rho<D$,
    contradicting the minimum distance of $\mathcal{C}$.
    Hence $\cksum_\rho(\bm m')\neq \cksum_\rho(\bm m'')$.
\end{proof}

\begin{lemma}[Reed--Solomon checksums]\label{lem:unique-decoding}
    Let $k,d \in \NN$ and let $\FF$ be a finite field with $\abs{\FF}\ge k$.
    Set $\cksumT \coloneqq 2d$.
    Then a $d$-unique decoding checksum function $\cksum_d: \FF^k \to \FF^{\cksumT}$ exists.
    With $\tO$ hiding $\polylog(k)$ factors,
    $\cksum_d: \FF^k \to \FF^\cksumT$ can be evaluated in $\tO(k \cksumT \Fbits)$ time,
    and its implementation circuit over $\FF$ has size $O(k\cksumT)$ and depth $\tO(1)$.
\end{lemma}
\begin{proof}
    If $2d < k$,
    let $H\in\FF^{2d\times k}$ be a parity-check matrix of a Reed--Solomon code over $\FF$ with block length $k$,
    dimension $k-2d$,
    and minimum distance $2d+1$.
    Define $\cksum_d(\bm m)\coloneqq H\bm m$.
    By the previous proposition, this is a $d$-unique decoding checksum function.

    If $2d\ge k$,
    let $\cksum_d$ be the identity map on $\FF^k$ padded with $2d-k$ zero coordinates.
    Then $\cksum_d$ is injective, so it is $d$-unique decoding.

    The evaluation bound follows by computing the $\cksumT$ linear forms defining $\cksum_d$.
\end{proof}

Generalizing the notion of checksums to matrices,
given a matrix $\bm M \in \FF^{k \times L}$ with columns $(\bm m_1,\ldots,\bm m_L)$,
we use
$\cksum_d(\bm M)$ to denote $\chi = (\cksum_d(\bm m_1),\ldots,\cksum_d(\bm m_L)) \in \FF^{\cksumT \times L}$.

\subsection{Interactive Protocols}\label{sec:prelim:IP}
An $(\ell, a, \Ptime, \Vtime, \Sigma)$-protocol is a \emph{public-coin,
    $\ell$-round interactive protocol, with alphabet  $\Sigma$,
    per-round message length $a$, and prover runtime $\Ptime$ and verifier runtime $\Vtime$}.
Specifically, such a protocol consists of a pair of interacting Turing machines $\prot[\pp]{}{(x)}{(y)}$,
where the $\cP$ has input a string $x \in \bin^*$,
and the $\cV$ has input a string $y \in \bin^*$,
and both parties have explicit access to the input parameters $\pp$.
We omit $x$, $y$, and $\pp$ from the notation when they are not important to the context.
The machines may also take in other private parameters as additional input
(in particular, the prover might have some extra information that makes it more efficient), 
but we omit them from the notation for simplicity.
The machine $\cP$ is deterministic and runs in time $\Ptime$, and is called \emph{the prover}, while $\cV$ is probabilistic, runs in time $\Vtime$, and is called \emph{the verifier}.
$\cP$ and $\cV$ are the \emph{two parties} of the protocol.
We omit the specification of $\Sigma$ when $\Sigma = \bin$.

In each round $j \in [\ell]$ of the protocol:
\begin{enumerate}
    \item $\cV$ sends a random message $q_j \gets_R \Sigma^a$ to $\cP$, referred to as a \emph{query}.
    \item $\cP$ responds with a message $a_j \in \Sigma^a$ determined by the prescribed next-message function, referred to as an \emph{answer}, which (abusing notation) is denoted by
          \[a_j \coloneqq \cP(x, j, (q_1, \ldots, q_j)) \in \Sigma^a.\]
\end{enumerate}
We make the simplifying assumption that both parties send messages of equal length,
and that $\cV$ never rejects in the middle of an execution.
We denote the sequence of verifier random coins by $\bm q \coloneqq (q_1,\ldots, q_\ell)$.
Finally,
$\cP(x, \bm q) \coloneqq (a_1, \ldots, a_\ell)$, where the $a_j = \cP(x, j, \bm q_{\le j})$ are the prescribed messages,
and the notation $\outputprotUU[\pp]{(x)}{(y; \bm q)} \in \bin^* \cup \set{\bot}$ represents the output of the verifier at the end of the interaction,
which contains a special symbol $\bot$ to indicate that the verifier rejected.
Note that the randomness is only over the random coins $\bm q$ of $\cV$,
and we omit $\bm q$ and $\pp$ from the notation when they are not important in the context.

The \emph{total communication complexity} of the protocol is the number of bits exchanged between the prover and the verifier, i.e. $2a\ell\cdot\log(\abs{\Sigma})$.

Let $\cL \subset \bin^*$ be a language.  We next define the notion of an
interactive proof of $\cL$.
\begin{definition}[$\epsilon$-Sound $\protcplx{}$-\IP]
    \label{def:IP}
    An $\protcplx{}$-protocol $\prot[\pp]{}{(x)}{(y)}$,
    where the public parameter is $\pp$ and $y = x$,
    is an \emph{interactive proof} (\IP) with soundness error $\epsilon$
    for a language $\cL$,
    if it satisfies the following completeness and soundness conditions.
    \begin{itemize}[leftmargin=*,label=-]
        \item \textbf{Completeness:}
              For any $x \in \cL$, there exists a prover strategy $\cP$ such that
              \[
                  \Pr[\lrag{\cP(x), \cV(x)}_\pp \neq \bot] = 1.
              \]

        \item \textbf{$\epsilon$-Soundness:}
              For any $x\notin\cL$ and any (computationally unbounded) prover strategy $\cP^*$,
              \[
                  \Pr[\lrag{\cP^*(x), \cV(x)}_\pp \neq \bot] \le \epsilon.
              \]
    \end{itemize}
\end{definition}

\subsubsection{Low-depth-predicate Interactive Proof and the \GKR protocol}\label{sec:LDP-IP}
We consider the following special type of protocol where the verifier does not query the input $x$ but instead outputs a low-depth predicate to be checked against the input.
\begin{definition}[$\epsilon$-Sound $\protcplx{}$-\LDP]
    \label{def:LDP}
    Let $\FF$ be a field.
    An \emph{low-depth-predicate-interactive-proof} (\LDP) with respect to a field $\FF$ for a language $\cL$ is a protocol
    with public parameters $\pp$, where the prover is given some input $x \in \bin^n$,
    and the verifier receives no input (except the protocol's parameters).
    At the end of the protocol, the verifier either rejects (outputs $\bot$) or outputs the description of a low-depth predicate $\Psi$,
    which takes as input $x$ and outputs 0 or 1,
    satisfying the following:
    \begin{itemize}[leftmargin=*,label=-]
        \item \textbf{Completeness:} If $x \in \cL$, there exists a prover strategy $\cP$ such that 
              \[ \Pr[\outputprotUU[\pp]{(x)}{} = \lrag \Psi \text{ s.t. } \Psi(x) = 1] = 1, \]
        \item \textbf{Soundness:} If $x \notin \cL$, then for any (computationally unbounded) prover strategy $\cP^*$,
              \[ \Pr[\outputprotUU[\pp]{^*(x)}{} = \lrag \Psi \text{ s.t. } \Psi(x) = 1] \le \epsilon. \]
    \end{itemize}
    Importantly, the verifier never accesses $x$ throughout the protocol.
\end{definition}

We restate the main result from \cite{JACM:GolKalRot15} in terms of \LDP.
\begin{theorem}[The \GKR Protocol, as an \LDP]
    \label{lem:GKR}
    Let $\FF$ be a field,
    and let $\C: \FF^n \to \FF$ be an 
    arithmetic circuit with addition and multiplication gates of fan-in 2 over $\FF$,
    with description $\lrag{\C}$.
    Let $G(x, \lrag \C)$ be the log-space uniform circuit that outputs $\C(x)$ on input $x \in \FF^n$ and $\lrag{\C} \in \bin^{\abs{\lrag{\C}}}$. 
    Denote the depth and size of $G$ by $D = D(n) \ge \log n$ and $S = S(n) \ge n$.

    For some constant $C_\GKR > 0$, there exists an $\epsilon$-sound $\protcplx{}$-\LDP,
    abbreviated as $\GKR \coloneqq \prot[\FF, \lrag \C]{}{(x)}{}$,
    for the language $\cL_{\C} = \{x \in \FF^n : \C(x) = 1\}$,
    with the following complexity (with $\tO$ ignoring poly-logarithmic factors in $D, \log S$):
    \begin{itemize}
        \item The soundness error is bounded by $\epsilon \le C_\GKR\cdot \left(\frac{D \log S}{\abs{\FF}}\right) = O(\frac{D \log S}{\abs{\FF}})$.
        \item $\ell = O(D \cdot \log S)$.
        \item $a = O(\Flog)$.
        \item $\Ptime = \tO(\poly(S)\cdot \Fbits)$.
        \item $\Vtime = \tO(D \log S\cdot \Flog  + \abs{\lrag{\C}}\cdot\Flog)$
              (and $\cV$ does not access $x$).
    \end{itemize}
    $\abs{\lrag \Psi} = O(\Flog)$, and the implementation circuit $C$ for $\Psi$ satisfies
    \begin{itemize}
        \item $\size(C) = \tO(n)$.
        \item $\depth(C) = \tO(1)$.
    \end{itemize}
\end{theorem}

\subsubsection{Interactive Proof of Proximity with Row Reduction}
An Interactive Proof of Proximity with Row Reduction
allows us to reduce checking a predicate $\C$ over a matrix $\bm M$ to checking a related predicate $\CRR$ over a subset of rows of $\bm M$.
It is a generalization of the standard \IPP,
which is defined over bit strings and with respect to Hamming distance.

\begin{definition}[$\epsilon$-Sound $\protcplx{}$-\IPP with Row Reduction]\label{def:ipp-row-red}
    Let $\FF$ be a field,
    $k, L \in \NN$,
    and $\C : \FF^{k \times L} \to \bin$ be a predicate.
    An $\epsilon$-sound $\protcplx{}$ \emph{Interactive Proof of Proximity} (\IPP) with Row Reduction for the language
    \[
        \cL_{\C} \coloneqq \set{\bm M \in \FF^{k \times L}: \C(\bm M) = 1},
    \]
    is an $\protcplx{}$-protocol,
    abbreviated as $\IPP \coloneqq \prot[\FF, \lrag \C, d]{}{(\bm M)}{}$ whose prover input is a matrix $\bm M \in \FF^{k \times L}$,
    and the verifier receives no input (except the protocol's parameters).
    At the end of the protocol, the verifier outputs either $\bot$ or $(\lrag \cQ, \lrag \CRR)$,
    where $\cQ \subsetneq [k]$ is a set of rows and $\CRR$ is a predicate such that the following holds.
    \begin{itemize}[label=-]
        \item \textbf{Completeness:} If $\C(\bm M) = 1$, then 
              $\Pr[\CRR({\bm M}[\cQ, :]) = 1] = 1$.
        \item $\epsilon$-\textbf{Soundness:}
              Suppose $\bm M$ is $d$-$\Deltac$-far from $\cL_\C$, \ie $\rowball(\bm M) \cap \cL_\C = \varnothing$,\footnote{Recall that $\rowball(\bm M)$ is the set of all matrices that are $\Deltac$-d-close to $\bm M$}
              then for any prover strategy $\cP^*$,
              \[
                  \Pr\left[\outputprotUU[\FF, \lrag \C, d]{^*(\bm M)}{} = (\lrag \cQ, \lrag \CRR) \text{ s.t. } \CRR({\bm M}[\cQ, :]) = 1\right] \le \epsilon.
              \]
    \end{itemize}
\end{definition}
For the \IPP to be non-trivial, we require $\abs{\cQ} \ll k$.
Such an \IPP exists by the following theorem, which generalizes Lemma~4 in \cite{FOCS:BGHK25}.
\begin{restatable}{theorem}{RRIPP}
    \label{thm:IPP}
    There exists a constant $c > 0$ such that
    for all $\secpar$, $d$, $k$, $L \in \NN$,
    and a field $\FF$,
    and any description $\lrag{\C}$ of a predicate $\C$ with implementation circuit $C$ whose size and depth are $S$ and $D$, respectively,
    if the following holds:
    \begin{itemize}
        \item $d = \Omega(\secpar\log k)$,
        \item $\abs{\FF} = \Omega(2^{\secpar} \cdot ((\secpar dL\log (\abs{\FF} k))^c + D\log S))$,
    \end{itemize}
    then there exists a $2^{-\secpar}$-sound \IPP with row reduction,
    abbreviated by $\IPP(\C) \coloneqq \prot[\secpar, \FF, \lrag \C, d]{}{(\bm M)}{}$,
    for the input matrix $\bm M \in \bin^{k \times L}$,
    and its output subset $\cQ \subset [k]$ satisfies:
    \[\abs{\cQ} \le \ceil{24\secpar \cdot \frac{k}{d}}.\]

    Let $\tO$ omit $\polylog(\abs{\FF}, k, L)$ factors. 
    The complexity of the protocol is as follows.
    \begin{itemize}
        \item $\ell = \tO(D\log S)$.
        \item $a = \tO(d\aboundIPP)$.
        \item $\Ptime = \poly(k, L, d, S, \Flog)$.
        \item $\Vtime = \tO(dL(D \log S + \abs{\lrag{\C}}) + \poly(d))$.
    \end{itemize}
    $\abs{\lrag{\cQ}} = \tO(\poly(d))$ and $\abs{\lrag{\CRR}} = \tO(\aboundIPP)$.
    Let $G$ and $C$ be the implementation circuits of $\cQ$ and $\CRR$ respectively.
    They satisfy the following.
    \begin{itemize}
        \item $\size(G) = \tO(\poly(d))$.
        \item $\depth(G) = \tO(1)$.
        \item $\size(C) = \tO(\abs{\cQ} \cdot L)$.
        \item $\depth(C) = \tO(1)$.
    \end{itemize}
    The specific requirements on $d$ and $\abs{\FF}$ are
    \begin{itemize}
        \item $d \ge \dboundIPP$,
        \item $\abs{\FF} \ge \FboundIPP$,
    \end{itemize}
    where $C_\GKR$ is the constant in \Cref{lem:GKR}.
\end{restatable}

The construction utilizes the \GKR protocol (\Cref{lem:GKR}) as well as a special \IPP for the \emph{polynomial valuation language} (\pval) with row reduction (Theorem~9 in \cite{FOCS:BGHK25}; c.f. \cite{TCC:RotRot20}).
For completeness,
we provide the corresponding definitions and the proof of \Cref{thm:IPP} in \Cref{sec:IPP-proof}.

\section{Doubly-Efficient Proof for Space-bounded Computation}
We state our main result as follows.
\begin{theorem}[Formal Statement of \Cref{thm:informal-deip-gen}]
    \label{thm:deip}
    For all $n, T = T(n) > n, S = S(n), \secpar = \secpar(n) \in \NN$,
    there exists a $2^{-\secpar}$-sound $\protcplx{}$-\IP for deciding any language in $\TS(T, S)$.
    The complexity of the protocol is as follows.
    \begin{itemize}
        \item $\ell = 2^{O(\sqrt{\log T})} \cdot \poly(\secpar, \log S)$.
        \item $a = 2^{O(\sqrt{\log T})} \cdot S \cdot \poly(\secpar, \log S)$.
        \item $\Ptime = \poly(S, T, \secpar)$.
        \item $\Vtime = 2^{O(\sqrt{\log T})} \cdot S^2 \cdot \poly(\secpar, \log S)$.
    \end{itemize}
\end{theorem}
Letting $T = n^{O(\log n)}$ in the above theorem, we obtain a doubly efficient interactive proof system for every $\mathsf{PSPACE}$ language decidable in time $n^{O(\log n)}$. 
\begin{corollary}[Formal Statement of \Cref{thm:informal-deip}]
    \label{corl:main}
    For all $n, T = T(n)$ such that $n < T < n^{O(\log n)}$, $S = S(n) \in \poly(n), \secpar = \secpar(n) \in \NN$,
    there exists a $2^{-\secpar}$-sound $\protcplx{}$-\IP for deciding any language in $\TS(T, S)$.
    The complexity of the protocol is as follows.
    \begin{itemize}
        \item $\ell = \poly(n, \secpar)$.
        \item $a = S \cdot \poly(n, \secpar)$.
        \item $\Ptime = \poly(S, T, \secpar)$.
        \item $\Vtime = S^2 \cdot \poly(n, \secpar)$.
    \end{itemize}
\end{corollary}

\subsection{\LDP for the Batch Language}
We prove our main theorems by constructing a \emph{Low-Depth-Predicate}-\IP, or \LDP (as defined in \Cref{sec:LDP-IP}, \Cref{def:LDP}), for verifying a batch of claims for deterministic time-$t$ computations, which we call transition claims.
Recall that an \LDP is a special protocol where the verifier never accesses the input throughout the protocol execution,
and only outputs a description of a low-depth predicate $\Psi$ about the input $\bm x$.
We note that our approach deviates from prior works 
\cite{STOC:ReiRotRot16,FOCS:BGHK25}, which constructed doubly efficient interactive proofs by first constructing protocols for verifying a batch of unambiguous interactive proofs.

Formally, fix a Turing machine $\cM$ with time complexity $T(n)$ and space complexity $S(n)$.
For any $t \le T(n)$, let the language $\cL_{t}$ consist of all pairs $(x_\tstart, x_\tend) \in \bin^{S(n)} \times \bin^{S(n)}$ such that $\cM$ transitions from configuration $x_\tstart$ to configuration $x_\tend$ in exactly $t$ steps.
Given a batch size parameter $k = k(n)\in\mathbb{N}$, define the batch language 
\[\cL_t^k \coloneq\{((x_{i, \tstart}, x_{i, \tend}))_{i \in [k]}\mid\forall i\in [k]~(x_{i, \tstart}, x_{i, \tend})\in\cL_t\}.\]
\begin{theorem}
    For all $n \in \NN$, $S = S(n), T = T(n)$, $t = t(n) \le T, k = k(n), \secpar = \secpar(n) \in \NN$, 
    there exists an upper bound
    \[
    \Lambda(T, k) \in 2^{O\left(\sqrt{ \log \binom{2\log T + \log k}{\log T} }\right)},
    \]
    
    and an $\epsilon$-sound $\protcplx{}$-\LDP, denoted as $\Batch(t, k)$, for the batch language $\cL_t^k$, where 
    \begin{itemize}
        \item $\epsilon = 2^{-\secpar}$.
        \item $\ell = \Lambda(T, k) \cdot \poly(\secpar)$.
        \item $a = \Lambda(T, k) \cdot S \cdot \poly(\secpar)$.
        \item $\Ptime = \poly(T, k, S, \secpar)$,
        \item $\Vtime = \Lambda(T, k)\cdot S^2\cdot \poly(\secpar)$.
    \end{itemize}
    Furthermore,
    the output $\lrag \Psi$ defines a low-depth predicate $\Psi$ on $\bm x$,
    whose implementation circuit $C_\Psi$ satisfies
    \begin{itemize}
        \item $\size(C_\Psi) = (k+\Lambda(T, k)) \cdot S \cdot \poly(\secpar)$.
        \item $\depth(C_\Psi) = \Lambda(T, k) \cdot \poly(\secpar)$.
    \end{itemize}
    Finally, the description length $\abs{\lrag \Psi} = O(\Lambda(T, k) \cdot \poly(\secpar) \cdot S)$.
    \label{thm:batch}
\end{theorem}

With \Cref{thm:batch}, our proof of \Cref{thm:deip} is straightforward. 

\begin{proof}[Proof of \Cref{thm:deip}]
    This follows from \Cref{thm:batch}.
    \begin{enumerate}
        \item Let $x$ be the input to $\cM$. $\cP$ first simulates $\cM(x)$ for $T$ steps and stores its tableau $\bm \tau \in \bin^{T \times S}$ to be used as auxiliary information to the sub-protocols.
        \item Both parties select the appropriate parameters, 
        then run the protocol given by \Cref{thm:batch},
        denoted as $\Batch(T, 1)$,
        which outputs a predicate description $\lrag \Psi$ to the verifier. 
        \item $\cV$ lets $\bm x = (x_{\tstart}, x_{\tend})$,
        where $x_{\tstart}$ denotes the start state with $x$ as input,
        and $x_{\tend}$ denotes the accept state of $\cM$.
        \item $\cV$ accepts iff $C_\Psi(\bm x, \lrag \Psi) = 1$.
    \end{enumerate}
    An \LDP implies a standard \IP when the verifier $\cV$ has access to the input $\bm x$ in the clear.

    Plugging in $k = 1$ into the bound $\Lambda(T, k)$, we have
    \[
    \Lambda(T, 1) = 2^{O\left(\sqrt{\log \binom{2\log T}{ \log T}}\right)} = 2^{O\left(\sqrt{\log T}\right)}.
    \]
    The resulting \IP has identical round-complexity $\ell$
    and per-round communication complexity $a$ as $\Batch(T, 1)$, and
    \begin{itemize}
        \item
            $\Ptime' = O(TS) + \Ptime_\LDP = \poly(S, T, \secpar)$,
        \item 
            $\Vtime' = O(\size(C_{\Psi})) + \Vtime_\LDP = 2^{O\left(\sqrt{\log T}\right)} \cdot S^2\cdot\poly(\secpar)$.
    \end{itemize}
    The additional terms are due to the prover's simulation of $\cM$ and the verifier's evaluation of the circuit $\Psi$.
\end{proof}

\section{Low-Depth-Predicate Interactive Proof $\Batch(t, k)$}
In this section, we prove \Cref{thm:batch}.
\paragraph{Road map of our construction}
Our protocol in \Cref{thm:batch} is constructed by recursion.
Let $\lambda \in \NN$ be some parameter.
Recall that the parties try to verify $\bm x \in \cL_t^{k}$.
Let us denote the protocol for checking this as $\Batch(t, k)$.
We shall try to reduce this to constructing some $\Batch(t', k')$ where either (1) $k' \ll k$ and $t' = t$, or (2) $k' = k$ and $t' \ll t$.\footnote{For (2), we ended up reducing to some $(t', k')$ where $k' > k$ and $t' < t$, as in \Cref{fig:tk-grid-midpoint}. This is not ideal, but enough for the recursion to terminate because the two recursive calls still define a partial order on the grid of $(t, k)$.}
We denote these two reductions by $R_1$ and $R_2$, 
as illustrated in \Cref{fig:tk-grid}.
If both $R_1$ and $R_2$ can be performed doubly efficiently, then we are done,
since we can just apply them sufficiently many times and reduce to the case where both $k$ and $t$ are small,
and output the final low-depth predicate defined on only those instances. 
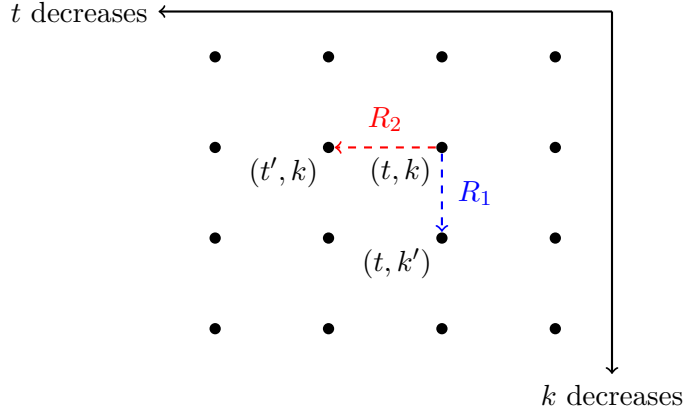
\begin{figure}[ht]
    \centering
    \begin{tikzpicture}[x=1.5cm, y=1.2cm]
        \foreach \t in {0,-1,-2,-3} {
            \foreach \k in {0,-1,-2,-3} {
                \node[circle, fill=black, inner sep=1.5pt] (P\t\k) at (\t, \k) {};
            }
        }
        
        \node[below left] at (P-1-1) {$(t, k)$};
        \node[below left] at (P-1-2) {$(t, k')$};
        \node[below left] at (P-2-1) {$(t', k)$};
        
        \draw[->, thick] (0.5, 0.5) -- (-3.5, 0.5) node[left] {$t$ decreases};
        \draw[->, thick] (0.5, 0.5) -- (0.5, -3.5) node[below] {$k$ decreases};
        
        \draw[->, dashed, blue, thick] (P-1-1) -- (P-1-2) node[midway, right, xshift=2pt] {$R_1$};
        \draw[->, dashed, red, thick] (P-1-1) -- (P-2-1) node[midway, above, yshift=2pt] {$R_2$};

    \end{tikzpicture}
    \caption{Grid of protocol dimensions $\Batch(t, k)$. Moving left decreases time $t$, and moving down decreases batch size $k$. The goal is to reduce to the case where both $t$ and $k$ are small.}
    \label{fig:tk-grid}
\end{figure}

When $t$ is tiny,
the $\GKR$ protocol applied to the circuit $\C(\bm x)$,
which verifies that all adjacent states are consistent with the transition rules of $\cM$ (see \Cref{clm:base-case}),
is of depth $\tO(t)$,
so we immediately obtain a doubly efficient \LDP.

The tricky case is when $t$ is not tiny, as $\depth(\C) = \Omega(t)$,
and the \GKR verifier is no longer efficient.
The good news is that we can reduce the number of instances $k$,  
by a careful application of an $\IPP$ with row reduction (Definition~\ref{def:ipp-row-red}). 
Before applying this protocol,
the verifier needs to prepare a claim $\Phi_\IPP$ with $\Deltac$-distance to $\bm x$ when $\bm x \notin \cL_t^{k}$.
Naively,
we can simply run $\Batch(t, k)$ to obtain the claim $\Phi_\IPP$
--- our \LDP is designed to output a false claim when $x \notin \cL_t^k$!
However,
this is a catch-22 as we are exactly trying to construct $\Batch(t, k)$.
Fortunately,
there is another \LDP that we can use to obtain such a claim when $x \notin \cL_t^k$,
which will become clear when we consider the other boundary case --- when $k$ is tiny,
hence let us consider this case.

When $k$ is too tiny for instance reduction and $t$ is still large, the prover can reduce $t$ by $\lambda$ by sending $\lambda$ many equally spaced ``midpoints'' between the length-$t$ path of $x_{\tstart}$ and $x_{\tend}$ for each length-$t$ transition statement $(x_{\tstart}, x_{\tend})$,
and both parties then need to apply the protocol $\Batch(t/\lambda, k \cdot \lambda)$.
Even though this \emph{increases} $k$ by a factor of $\lambda$,
we are still making progress because in $\Batch(t/\lambda, k \cdot \lambda)$, we can reduce the $(k \cdot \lambda)$-dimension by applying an $\IPP$ with row reduction again.

Going back to the case of general $(t, k)$,
we can apply a similar idea when we try to generate $\Deltac$ distance.
After putting the $k$ statements on $k$ rows, 
the prover introduces $\lambda - 1$ midpoint states for each statement,
hence creating the midpoint matrix $\Mx$ (\Cref{def:midpoint-matrix}).
However,
if the prover sends this matrix in the clear,
prohibitively many --- $\Omega(k \cdot \lambda)$ --- states have to be sent.
The good news is that
the verifier can enforce a joint constraint $\Phi_\IPP$ on the rows of $\Mx$,
exactly because they can reduce to $\Batch(t/\lambda, k \cdot \lambda)$ on $\xmid \in \bin^{(k \cdot \lambda) \times (2S)}$,
where each row in $\xmid$ 
corresponds to a $t/\lambda$-length transition statement specified in $\Mx$.
(Note that the symbol $\nwarrow$ refers to the input dimensions $(t/\lambda, k \cdot \lambda)$, which is situated in the upper-right corner of $(t, k)$ in \Cref{fig:tk-grid-midpoint}.)
\begin{figure}[ht]
    \centering
    \begin{tikzpicture}[x=3cm, y=1.2cm]
        \foreach \t in {0,-1,-2} {
            \foreach \k in {0,-1,-2} {
                \node[circle, fill=black, inner sep=1.5pt] (P\t\k) at (\t, \k) {};
            }
        }
        
        \node[below left] at (P-1-1) {$\Batch(t, k), \bm x$};
        \node[below left] at (P-1-2) {$\Batch(t, k/d), \xred$};
        \node[below left] at (P-20) {$\Batch(\frac{t}{\lambda}, k \cdot \lambda), \xmid$};
        
        \draw[->, thick] (0.5, 0.5) -- (-3.5, 0.5) node[left] {$t$ decreases};
        \draw[->, thick] (0.5, 0.5) -- (0.5, -3.5) node[below] {$k$ decreases};
        
        \draw[->, dashed, blue, thick, bend right=15] (P-1-1) -- (P-1-2) node[midway, right, xshift=2pt] {$R_1$};
        \draw[->, dashed, red, thick, bend left=15] (P-1-1) -- (P-20) node[midway, above, yshift=2pt] {$R_2$};

    \end{tikzpicture}
    \caption{Grid of the sub-protocol dimensions that we actually recur to, and their respective batches of statements. The two arrows represent the two reduced protocols that we actually construct recursively.}
    \label{fig:tk-grid-midpoint}
\end{figure}
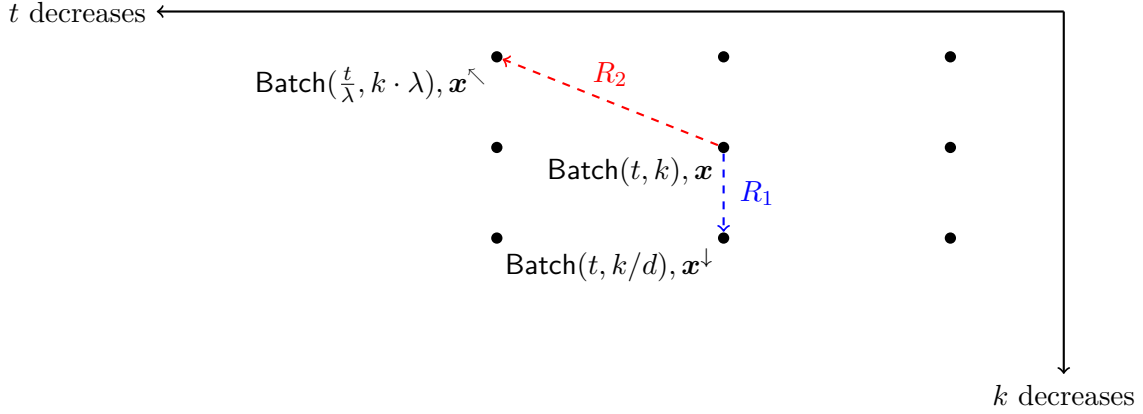
Therefore,
for some distance parameter $d = d(n)$,
they proceed with the following two steps:
\begin{enumerate}
    \item The prover only ``commits'' to these states by only sending the checksums $\chi$ of every column. 

    These checksums are constructed from a distance-$\tO(d)$ error correcting code,
    and force the prover to cheat on many rows if it cheats, as in previous literature.
    \item Both parties then run $\Batch(t/\lambda, k \cdot \lambda)$ on $\xmid$.
\end{enumerate}
The verifier then obtains a claim $\Phi_\IPP$ that is false on $\xmid$ (equivalently, $\Mx$) whenever $\xmid \notin \cL_{t/\lambda}^{k \cdot \lambda}$.
With this additional constraint,
the parties proceed with an $\IPP$ with row reduction for the language $\mathcal{L}_{\Phi_\IPP}$ on prover input $\Mx$. 
This protocol outputs a subset of roughly $(k/d)$ rows, $\cQ \subset [k]$, as well as a linear constraint $\CIPP$ over the sub-matrix $\Mred = \Mx[\cQ,:]$.
(Note that $\downarrow$ refers to $(t, k/d)$, which is situated beneath $(t, k)$ in \Cref{fig:tk-grid-midpoint}.)
Let $\xred$ be the $\abs{\cQ}$ length-$t$ statements specified by $\Mred$'s leftmost and rightmost columns.
\Cref{lem:DcRR} guarantees that if the original $\bm x \notin \cL_t^{k}$, 
then either $\xred \notin \cL_t^{k/d}$ or $\CIPP(\Mred) = 0$.
Owing to this additional linear constraint $\CIPP$,
which is defined on the reduced midpoint matrix $\Mred$,
instead of just checking $\xred \in \cL_t^{k/d}$,
we need to additionally check that the corresponding $\Mred$ satisfies $\CIPP(\Mred) = 1$.
This turns out to be straightforward by having $\Batch(t, k/d)$ perform this check on top of checking $\xred \in \cL_t^{k/d}$ throughout.

Finally, a minor technicality is that we have to ensure the resulting low-depth predicate is indeed defined over $\bm x$,
and not just over a related string.
This happens when we call $\Batch$ on dimensions other than $(t, k)$,
which we handle by redefining predicates over $\bm x$ as follows:
\begin{enumerate}
    \item When $k$ is large,
    $\Batch(t, k/d)$ returns a low-depth predicate over $\xred \in \bin^{(k/d) \times (2S)}$.

    The actual output predicate over $\bm x$ first selects $\xred = \bm x[\cQ, :]$ from the input $\bm x$ and checks that $\xred$ satisfies the low-depth predicate returned by $\Batch(t, k/d)$.

    \item Conversely, when $k$ is tiny,
    $\Batch(t/\lambda, k \cdot \lambda)$ returns a low-depth predicate over the midpoint statements $\xmid \in \bin^{(k \cdot \lambda) \times (2S)}$.

    The output predicate over $\bm x$ has the midpoint statements $\xmid$ hard-coded,
    and verifies that $\bm x$ is consistent with $\xmid$ on the boundaries and $\xmid$ indeed satisfies the given low-depth predicate.
\end{enumerate}
To sum up the case of general $(t, k)$,
our goal of constructing $\Batch(t, k)$ is reduced to constructing two sub-protocols (\Cref{fig:tk-grid-midpoint}):
\begin{enumerate}[label=(\alph*)]
    \item Calling $\Batch(t/\lambda,k \cdot \lambda)$ on $\xmid$.
    \item Calling $\Batch(t, k/d)$ on $\xred$.
\end{enumerate}
In addition to these recursive calls, 
the other steps, 
including a call to the $\IPP$ with row reduction (\Cref{thm:IPP}),
incur a fixed $\poly(\lambda, d, S)$ cost for the verifier.
To optimize the efficiency guarantees, we set \[
\lambda \coloneqq \ceil{2^{\sqrt{\log \binom{2\log T + \log k}{\log T}}}},d, \basek \in \poly(\lambda).\] 
Here we give some intuition for the efficiency bound for the case when $k = 1$ and $\lambda = 2^{\sqrt{\log T}}$.
The full analysis is deferred to \Cref{sec:complexity}.

Let us focus on the running time of the verifier, which is given by the recurrence $\mathsf{Vtime}(t, k) = \Vtime(t/\lambda,k\cdot \lambda) + \Vtime(t,k/\lambda) + \poly(S, \lambda)$. 
For simplicity,\footnote{The actual set of parameters is different; see \Cref{alg:main}.} we derive a closed-form upper bound on $\mathsf{Vtime}(T, 0)$ when $\baset = \basek = d = \lambda$ and $T = \lambda^\rho$.
Write $t = \lambda^i$ and $k = \lambda^j$, and let $\mathsf{Vtime}'(i, j) = \mathsf{Vtime}(\lambda^i, \lambda^j)$. 
Then we can verify that $\mathsf{Vtime}'$ satisfies the following: if $i = 0$ (the base case), $\Vtime'(i, j) = \tO(\lambda)$.
Otherwise, if $j = 0$, $\Vtime'(i, j) = \Vtime'(i-1, j)+\poly(S,\lambda)$, 
and if $j \ge 1$, $\Vtime'(i, j) = \Vtime'(i-1, j+1) + \Vtime'(i, j-1) + \mathsf{poly}(S, \lambda)$. 
Consider the recursion tree for $\mathsf{Vtime'}$ rooted at node $(\rho, 0)$. Our goal is to upper bound the number of nodes in this recursion tree, which would yield an upper bound on $\Vtime(T, 0)$. 

In this recursion tree, every node $(i, j)$ that has more than one child has exactly two children, which correspond to two moves that we can make: a $\nwarrow$-move from $(i, j)$ to $(i-1, j+1)$, and a $\downarrow$-move from $(i, j)$ to $(i, j-1)$. Starting at node $(\rho, 0)$, we must make exactly $\rho$ $\nwarrow$-moves to reach a node $(0, j)$, i.e., the base case. Since after making $\rho$ $\nwarrow$-moves we are at node $(0, \rho)$, and we can never reach a node $(i, j)$ such that $j < 0$, this implies that we can make at most $\rho$ $\downarrow$-moves. Suppose at some step we have made exactly $a$ $\nwarrow$-moves and $b$ $\downarrow$-moves; then we are at node $(\rho-a, a-b)$, and since we never reach a node $(i, j)$ such that $j < 0$, this implies that $a$ and $b$ must satisfy $a-b\ge 0$. The number of paths starting from $(\rho, 0)$ that make $\rho$ $\nwarrow$-moves and $\rho$ $\downarrow$-moves and at every step satisfy $a-b\ge 0$, is exactly the number of Dyck paths of size $\rho$, which is given by the $\rho$th Catalan number $C_\rho = {2\rho\choose \rho}/(\rho+1)$. Since the number of nodes in the recursion tree is at most the number of such paths, we have that ${2\rho\choose \rho}/(\rho+1) \cdot \poly(\lambda, S) \le 2^{2\rho} \cdot \poly(S)$ is an upper bound on $\Vtime(T, 0)$.
Since $\lambda = 2^{\sqrt{\log T}}$, $\rho = \log_\lambda T = \sqrt{\log T}$, 
and thus $\Vtime(T, 0) = 2^{O(\sqrt{\log T})} \cdot \poly(S)$.

\subsection{Our Construction}
We reduce the number of transition claims $k$ by iteratively applying the \IPP with row reduction (\Cref{thm:IPP}). 
To ensure soundness when sub-sampling to a smaller batch of transitions, we augment the batch language with a predicate $\C$.
This predicate $\C$ enforces auxiliary constraints --- such as consistency between the midpoints and the succinct checksums --- that ensure the smaller batch remains false if the original batch was false.
To implement this, we introduce the \emph{midpoint matrix}.
\begin{definition}[Midpoint Matrix]
    Let $\bm x = ((x_{i, \tstart}, x_{i, \tend}))_{i \in [k]}$ be an instance of $\mathcal{L}_t^k$ for some $t\ge\lambda$, where $\lambda\in\mathbb{N}$ is some splitting parameter. For all $i\in [k]$, $j\in \{0, \ldots, \lambda\}$, let $x_{i, j}$ denote the configuration that $\mathcal{M}$ reaches after $(\frac{t}{\lambda})\cdot j$ time steps when starting from configuration $x_{i,0} = x_{i, \mathsf{start}}$. Define the \emph{midpoint matrix} ${\bm M} = {\bm M}_{\bm x}\in\{0, 1\}^{k\times (S(\lambda+1))}$ as follows: for all $i\in [k]$, $j\in\{0, \ldots, \lambda\}$, ${\bm M}[i, j] = x_{i, j}$. 
    \label{def:midpoint-matrix}
\end{definition}
And the augmented language is
\[
\cL_{t}^{k}[\C]\coloneq\{{\bm x} \in \bin^{k \times 2S}\mid \bm x \in \cL_t^k\land (t\ge\lambda\implies\C({\bm M}_{\bm x}) = 1)\}.
\]
Initially, we set $\lrag {\C_0} = \top$ to denote the fact that there is no additional constraint (and this predicate always outputs 1).

We show that $\cL_{t}^{k}[\C]$ admits a doubly efficient \LDP, 
which we call $\Batch(t, k)$.
\begin{theorem}[Protocol $\Batch(t, k)$, an \LDP for the augmented language]
    \label{thm:main-LDP}
    For $n\in\mathbb{N}$, let $S = S(n)$, $T = T(n) > n$, and $\secpar = \secpar(n)$. Suppose $t = t(n) \le T$ and $k = k(n)$.
    Let $C$ be the implementation circuit for a predicate $\C:\bin^{k \times (S(\lambda + 1))} \to \bin$.

    For some upper bound
    \[
    \Lambda(T, k) \in 2^{O\left(\sqrt{\log \binom{2\log T + \log k}{\log T}}\right)},
    \]
    there exists an $\epsilon$-sound $\protcplx{}$-\LDP (\Cref{alg:main}),
    abbreviated as $\Batch(t,k) \coloneqq \prot[\BatchParam]{}{(\bm x)}{}$ for the language $\cL_{t}^{k}[\C]$.
    We additionally assume $\cP$ has access to the entire tableau of the computation $\bm \tau \in \bin^{T \times S}$ as auxiliary input.
    The complexities of the protocol are as follows.
    Note that the bounds here are loose and use the fact that $t \le T$. 
    \begin{itemize}
        \item $\epsilon \le 2^{-\secpar}$.
        \item $\ell = (\Lambda(T, k) + \depth(C) \log \size(C)) \cdot \poly(\secpar)$.
        \item $a = \Lambda(T, k) \cdot S \cdot \poly(\secpar)$.
        \item $\Ptime = \poly(S, k, T, \secpar, \size(C))$.
        \item $\Vtime = (\Lambda(T, k) + \depth(C) \log\size(C) + \abs{\lrag \C}) \cdot S^2 \cdot \poly(\secpar) $.
    \end{itemize}
    Furthermore,
    the output $\lrag \Psi$ defines a low-depth predicate $\Psi$ on $\bm x$,
    whose implementation circuit $C_\Psi$ satisfies
    \begin{itemize}
        \item $\size(C_\Psi) = (k+\Lambda(T, k)) \cdot \poly(\secpar) \cdot S$.
        \item $\depth(C_\Psi) = \Lambda(T, k) \cdot \poly(\secpar)$.
    \end{itemize}
    Finally, $\abs{\lrag \Psi} = \Lambda(T, k) \cdot \poly(\secpar) \cdot S$.
\end{theorem}
\Cref{thm:batch} follows from \Cref{thm:main-LDP} by letting $\lrag{\C} = \top$,
defined to be the predicate that always outputs 1,
in which case $\cL_t^k[\C] = \cL_t^k$ and the terms involving $C$ and $\abs{\lrag{\C}}$ are all $O(1)$.

\paragraph{The protocol}
\label{sec:construction}
For a summary of notation, see \Cref{table:symbols}.
\begin{table}[htbp]
    \centering
    \begin{tabular}{|c|c|c|}
    \hline
    \textbf{Symbol} & \textbf{Meaning} & \textbf{Remark} \\
    \hline
    $S$ & TM space bound & \NA\\
    \hline
    $k$ & Batch size & \NA \\
    \hline
    $t$ & Transition length & $\le T$\\
    \hline
    $(x_{i, \tstart}, x_{i, \tend})$ & $i$-th transition claim & $i \in [k]$\\
    \hline
    $\bm x = ((x_{i, \tstart}, x_{i, \tend}))_{i \in [k]}$ & Input batch & $\bm x \in \bin^{2Sk}$\\
    \hline
    $\lambda$ & Splitting parameter & $\Lambda(T, k)$\\
    \hline
    $x_{i,j}$ & TM state after $(\frac{t}{\lambda}) \cdot j$ steps from $x_{i, \tstart}$ & $j \in [0, \lambda]$\\
    \hline
    $\bm M = (x_{i,j})_{i \in [k], j \in [0, \lambda]}$ & Midpoint matrix & $x_{i, 0}=x_{i, \tstart}, x_{i, \lambda}=x_{i, \tend}$\\
    \hline
    $\secpar$ & Soundness parameter & \NA \\
    \hline
    $\FF$ & A binary field & specified in \Cref{alg:main}\\
    \hline
    $\chi = \cksum_d(\bm M)$ & Checksum of the midpoint matrix & $\chi \in \FF^{\cksumT \times S(\lambda+1)}$\\
    \hline
    $\tmid = t/\lambda$ & Length of sub-transitions & \NA\\
    \hline
    $\kmid = k\lambda$ & Expanded batch size & \NA\\
    \hline
    $\xmid = (x_{i,j})_{(i, j) \in [k] \times [0, \lambda - 1]}$ & Expanded batch & $\xmid \in \bin^{\kmid \times 2S}$\\
    \hline
    $\kred = |\cQ|$ & Reduced batch size & $\approx k/d$\\
    \hline
    $\xred = \bm x [\cQ, :]$ & Reduced batch & $\xred \in \bin^{\kred \times 2S}$\\
    \hline
    \end{tabular}
    \caption{\textbf{Summary of Symbols.} }
    \label{table:symbols}
\end{table}
Let $\lambda = \poly(n)$ be a parameter.
Without loss of generality, we assume $t = \lambda^\logt$ for some $\logt \in \NN$.
Let
\[
\recnodes \coloneqq
\binom{2\log_\lambda T+\log_\lambda(k+1)+O(1)}
      {\log_\lambda T+O(1)}
\]
denote the recurrence-tree node bound used in the soundness and complexity analyses, and set
\[
\secloc \coloneqq \secpar+\ceil{\log(4\recnodes)}+2.
\]
For some absolute constants $c_0, c_1$ to be analyzed in \Cref{sec:field-size},
we let $\abs{\FF}$ be the smallest binary field such that $\abs{\FF} \ge \Fboundstmt$.

For the base case when $k < \basek$,
apply the following \SmallBatch protocol.
\begin{lemma}[Base case protocol for small batch size]
\label{lem:base-case}
    For $n\in\mathbb{N}$, let $S = S(n)$, $T = T(n) > n$, and $\secpar = \secpar(n)$. Suppose $t = t(n) \le T$, $k < \basek$, and $\lambda = \Lambda(T, k)$.
    Let $C$ be the implementation circuit for the predicate $\C:\bin^{k \times (S(\lambda + 1))} \to \bin$.

    There exists an $\epsilon$-sound $\protcplx{}$-\LDP,
    abbreviated as $\SmallBatch \coloneqq \prot[\BatchParam]{}{(\bm x)}{}$ for the language $\cL_{t}^{k}[\C]$.
    We additionally assume $\cP$ has access to the entire tableau of the computation $\bm \tau \in \bin^{T \times S}$ as auxiliary input.
    The complexities of the protocol are identical to the ones stated in \Cref{thm:main-LDP}.
\end{lemma}

We present the two protocols, \SmallBatch and \Batch,
in \Cref{alg:base,alg:main}, respectively.
In \SmallBatch,
    the prover computes the midpoint matrix,
$\bm M \in \bin^{k \times S(\lambda+1)}$, 
but given that $k$ is small,
the prover can afford to send it in the clear to the verifier.
They then make the recursive call to $\Batch(t/\lambda, k \cdot \lambda)$ on the expanded batch $\xmid$.
At the end, the verifier performs the resulting checks explicitly on $\bm M$.
The general $\Batch(t, k)$ protocol in \Cref{alg:main} uses the auxiliary circuits defined in \cref{clm:base-case,clm:small-batch,clm:reduction-case,clm:batch}.

\begin{algorithm}[htbp]
  \setstretch{1.2}
  \caption{Protocol $\Batch(t, k) = \prot[S, T, t, k, \lrag{\C}]{}{(\bm x)}{}$ for $\cL_t^{k}[\C]$.}
  \label{alg:main}
    \textbf{Input Parameters:} $t, k, T \in \NN$,
    $\lrag \C$ is the description of a predicate $\C : \bin^{k \times S (\lambda + 1)} \to \bin$ with implementation circuit $C$.\\
    \textbf{Input Batch:} $\bm x \in \bin^{2S \cdot k}$.\\
    \textbf{Derived Parameters:} $\lambda = \Lambda(T, k)$, $\recnodes$ and $\secloc$ as above, $\baset = \lambda$, $d = \dbound$, $\basek = d^2$,
    $\FF$ is the smallest binary field such that $\abs{\FF} \ge \Fboundstmt$ for constants $c_0, c_1$ analyzed in \Cref{sec:field-size}. These parameters are \emph{global parameters}: they only depend on $(T, k)$ in the root call to $\Batch(T, k)$.\\
    \textbf{Verifier Output:} $\lrag \Psi \in \bin^*$, describing a predicate $\Psi: \bin^{2S \cdot k} \to \bin$.
    \begin{enumerate}[label=(\arabic*)]
        \setcounter{enumi}{-1}
        \item \textbf{Base Case Handling}

        \begin{algorithmic}
            \If{$t < \baset$}
                \State Let $\Cbase$ be the low-depth circuit that checks $\bm x \in \cL_t^k$ (\Cref{clm:base-case}).
                \State Apply $\GKR\printparam{\FF, \lrag \Cbase}$ (\Cref{lem:GKR}) and \Return its output.
            \ElsIf{$k < \basek$}
                \State Call $\SmallBatch\printparam{\BatchParam}$ (\Cref{lem:base-case}) and \Return its output.
            \EndIf
        \end{algorithmic}

        \item \textbf{Midpoint Expansion \& Checksum} 
        
        $\cP$ finds the matrix $\bm M \in \bin^{k \times S(\lambda + 1)}$ containing all the midpoints from the tableau of $\cM$.\\
        $\cP$ sends the checksum $\chi = \cksum_d(\bm M)$ to $\cV$.
        \label{ln:tableau}

        \item \textbf{Recursive Call on the Expanded Batch} 
        
        Let $\xmid$ be the expanded batch of claims, of dimension $(\tmid, \kmid) \coloneqq (\frac{t}{\lambda}, k\lambda)$ in $\bm M$.\\
        Call $\Batch(\tmid, \kmid) = \prot[\FF, n, \tmid, \kmid, \top]{}{(\xmid)}{}$ on the expanded instance $\xmid$ to ensure that all intermediate states are locally consistent,
        to obtain $\lrag \Cmid$.
        \label{ln:call-exp}

        \item \textbf{Interactive Proof of Proximity with Row Reduction} 
        
        Let $\Creduce$ be the circuit that verifies $\cksum_d(\bm M) = \chi$, $\Cmid(\xmid(\bm M)) = 1$ and $\C(\bm M)=1$.\\
        Apply $\IPP\printparam{\secloc, \FF, d, \lrag \Creduce}$ (\Cref{thm:IPP}) to reduce checking $\Creduce(\bm M)$ to $\CIPP(\bm M[\cQ, :])$.
        \label{ln:RR}

        \item \textbf{Recursive Call on the Reduced Batch \& Post-Processing} 

        Let $\xred \coloneqq \bm x[\cQ, :]$.\\
        Call $\Batch(t, \kred) = \prot[\FF, n, t, \abs{\cQ}, \lrag \CIPP]{}{(\xred)}{}$ on $\xred = \bm x[\cQ, :]$ to obtain $\lrag{\Cred}$\\
        \textbf{return} $\lrag{\CBatch}$, the description of the circuit in \Cref{clm:batch} that checks $\cksum_d(\bm x[:,0:S]) = \chi[:,0:S], \cksum_d(\bm x[:, -S:])=\chi[:,-S:]$ and that $\Cred(\bm x[\cQ, :]) = 1$.
        \label{ln:red}
    \end{enumerate}
\end{algorithm}

\begin{algorithm}[htbp]
  \setstretch{1.2}
  \caption{Protocol $\SmallBatch = \prot[S, T, t, k, \lrag \C]{_\SmallBatch}{(\bm x)}{}$ for $k < \basek$.}
  \label{alg:base}
      \textbf{Input Parameters:} $t, k, T \in \NN$, 
      $\lrag \C$ is the description of a predicate $\C : \bin^{k \times S (\lambda + 1)} \to \bin$ with implementation circuit $C$.
    \textbf{Input Batch:} $\bm x \in \bin^{2S \cdot k}$.\\
    \textbf{Derived Parameters:} $\lambda = \Lambda(T, k)$, $\recnodes$ and $\secloc$ as above, $\baset = \lambda$, $d = \dbound$, $\basek = d^2$, $\FF$ is the smallest binary field such that $\abs{\FF} \ge \Fboundstmt$ for constants $c_0, c_1$ analyzed in \Cref{sec:field-size}.\\
    \textbf{Verifier Output:} $\lrag \Psi \in \bin^*$, describing a predicate $\Psi: \bin^{2S \cdot k} \to \bin$.\\
    \textbf{Assumption:} $k < \basek$ is small.
    \begin{enumerate}[label=(\arabic*)]
        \item \textbf{Midpoint Expansion}
        
        $\cP$ finds the matrix $\bm M \in \bin^{k \times S(\lambda + 1)}$ containing all the midpoints and sends $\bm M$ to $\cV$.

        \item \textbf{Recursive Call on the Expanded Batch} 
        
        Call $\Batch( \tmid, \kmid) = \prot[\FF, n, \tmid, \kmid, \top]{}{(\xmid)}{}$ on the expanded instance $\xmid$ to obtain $\lrag{\Cmid}$.

        \item \textbf{Explicit Checks and Post-processing}
        \label{ln:CRR-basek}

        $\cV$ checks $\CCmid(\xmid, \lrag{\Cmid}) = 1$ and $\C(\bm M)=1$. It \textbf{rejects} otherwise.\\
        $\cV$ constructs $\xmid$ from $\bm M$. \\
        \textbf{return} $\lrag{\CSmallBatch}$, the description of the circuit in \Cref{clm:small-batch} that verifies the boundaries of $\xmid$ agree with $\bm x$ and that $\Cmid(\xmid) = 1$.
    \end{enumerate}
\end{algorithm}

\begin{proposition}[The input predicate $\Cbase$ for $\cL_t^{k}$]
\label{clm:base-case}
    Consider parameters $n, t = t(n), k = k(n), S = S(n) \in \NN$, and a batch of statements $\bm x \in \bin^{k \times (2S)}$.
    The predicate $\Cbase$ with description $\lrag\Cbase = (t, k)$ verifies that $\bm x \in \cL_t^{k}$.
    The implementation circuit $\CCbase$ for $\Cbase$ satisfies
    \begin{itemize}
        \item $\size(\CCbase) = O(t \cdot k \cdot S)$.
        \item $\depth(\CCbase) = O(t + \log k + \log S)$.
    \end{itemize}
\end{proposition}
\begin{proposition}[The output predicate $\CBatch$ of $\Batch$]
    \label{clm:batch}
    Consider parameters $n, k, S = S(n), d = d(n), \lambda = \lambda(n) \in \NN$,
    $\cksumT = 2d$,
    and a boundary checksum slice $\cksumv_{\mathsf{bdry}} \in \FF^{\cksumT \times 2S}$.
    Let $\cQ \subset [k]$ be a subset of indices, and let its implementation circuit be $G$.
    Given an instance $\bm x \in \bin^{k \times (2S)}$ of $\cL_t$,
    a predicate description $\lrag \Cred$ for $\Cred : \bin^{\abs{\cQ} \times (2S)} \to \bin$ whose implementation circuit is $\CCred$,
    the predicate $\CBatch$ with description $\lrag\CBatch = (\cksumv_{\mathsf{bdry}}, \lrag \cQ, \lrag \Cred)$ (where $\Cred : \bin^{\abs{\Q} \times (2S)} \to \bin$) checks the following.
    \begin{enumerate}
        \item $\cksum_d(\bm x) = \cksumv_{\mathsf{bdry}}$.
        \item Expand $\langle\cQ\rangle$ using $G$ and compute $\xred \coloneqq \bm x[\cQ, :]$.
        \item Verify that $\Cred(\xred) = 1$.
    \end{enumerate}
    Note that $\abs{\lrag \CBatch} = \abs{\lrag \Cred} + \abs{\lrag \cQ} + O(\cksumT S\Flog) = \abs{\lrag \Cred} + \abs{\lrag \cQ} + O(dS\Flog)$.
    With $\tO$ hiding $\polylog(n, S k, d, \lambda)$ factors,
    the implementation circuit $\ensuremath{C_{\Batch}}\xspace$ for $\ensuremath{\Psi_{\Batch}}\xspace$ satisfies
    \begin{itemize}
        \item $\mathsf{size}(\CCBatch) = \tO(k\cksumT S +\abs{\cQ}(\size(G) + S) + \size(\CCred))$.
        \item $\mathsf{depth}(\CCBatch) = \tO(1) + \depth(G) + \depth(\CCred)$.
    \end{itemize}
\end{proposition}

\begin{proposition}[The input predicate $\Creduce$ for the \IPP]
    \label{clm:reduction-case}
    Consider parameters $n, S = S(n), k = k(n), d = d(n), \lambda = \lambda(n)\in \NN$.
    Given a checksum $\cksumv \in \FF^{\cksumT \times (S(\lambda+1))}$, a predicate description $\lrag \C$ for $\C: \FF^{k \times (S(\lambda+1))} \to \bin$ with implementation circuit $C$, 
    a matrix $\bm M \in \bin^{k \times (S(\lambda+1))}$, 
    and a predicate description $\lrag{\Cmid}$ for $\Cmid: \bin^{(k \lambda) \times (2S)} \to \bin$ with implementation circuit $\CCmid$,
    the predicate $\Creduce$ with description $\lrag\Creduce = (\cksumv, \lrag \C, \lrag{\Cmid}, d, \lambda)$ performs the following on input $\bm M$:
    \begin{enumerate}
        \item Reads out $\xmid = ((\xmid_{i,j}, \xmid_{i,j + 1}))_{i, j \in [k] \times [0, \lambda - 1]} \in \bin^{(k \cdot \lambda) \times (2 S)}$ from $\bm M$ and verifies that $\Cmid(\xmid) = 1$.
        \item Verifies that $\cksum_d(\bm M) = \cksumv$.
        \item Verifies that $\C(\bm M) = 1$.
    \end{enumerate}
    With $C$ denoting $\C$'s implementation circuit and $\tO$ hiding $\polylog(n, S, k, d, \lambda)$ factors, the implementation circuit $\CCreduce$ for $\Creduce$ satisfies
    \begin{itemize}
        \item $\size(\CCreduce) = \tO(k\cksumT S\lambda +\size(C) + \size(\CCmid))$.
        \item $\depth(\CCreduce) = \tO(\max(\depth(C), \depth(\CCmid)))$.
    \end{itemize}
\end{proposition}

\begin{proposition}[The output predicate $\CSmallBatch$ of $\SmallBatch$]
    Consider parameters $n, S = S(n), k = k(n), \lambda = \lambda(n) \in \NN$.
    Let $\bm x \in \bin^{k \times (2S)}$ be a batch of statements and \hfill\break
    $\xmid = ((x_{i, j}, x_{i, j + 1}))_{(i, j) \in [k] \times [0, \lambda - 1]} \in \bin^{(k \cdot \lambda) \times (2S)}$ be the batch of midpoint statements for $\bm x$ (given by the midpoint matrix ${\bm M}_{\bm x}$ for parameter $\lambda$).
    Given a description $\Cmid$ for a predicate $\Cmid : \bin^{(k \lambda) \times (2S)} \to \bin$ whose implementation circuit is $\CCmid$,
    we let $\lrag\CSmallBatch = (\xmid, \lrag{\Cmid})$. 
    The predicate $\CSmallBatch(\bm x)$ verifies the following:
    \begin{enumerate}
        \item The boundary states of $\xmid$ agree with $\bm x$ (i.e., $x_{i,0} = x_{i, \tstart}$ and $x_{i, \lambda} = x_{i, \tend}$ for all $i \in [k]$).
        \item The condition $\Cmid(\xmid) = 1$ holds.
    \end{enumerate}

    Note that $\abs{\lrag \CSmallBatch} = O(k \lambda S + \abs{\lrag{\Cmid}})$.
    With $\tO$ hiding $\poly(n,S,k,\lambda)$ factors,
    the implementation circuit $\CCSmallBatch$ for $\CSmallBatch$ satisfies
    \begin{itemize}
        \item $\size(\CCSmallBatch) = \tO(kS\lambda) + \size(\CCmid)$.
        \item $\depth(\CCSmallBatch) = \tO(1) + \depth(\CCmid)$.
    \end{itemize}
    \label{clm:small-batch}
\end{proposition}

\subsection{Analysis of Our Construction}
We show how to prove \Cref{thm:main-LDP,lem:base-case} simultaneously,
assuming \Cref{clm:base-case,clm:reduction-case,clm:small-batch,clm:batch},
whose proofs are deferred to \Cref{sec:auxiliary-claims}.

By strong induction,
suppose for all $(t', k')$ such that either $t' < t$ or $t' = t$ and $k' < k$,
there exists a protocol $\Batch(t', k')$ with the stated properties. 

\subsubsection{Completeness and parameter selection} 
\label{sec:field-size}
Completeness follows from the completeness of the underlying sub-protocols.

We set $\lambda = \ceil{2^{\sqrt{\log \binom{2\log T + \log k}{\log T}}}}$, $\secloc = \secpar+\ceil{\log(4\recnodes)}+2$, $d = \dbound$, $\baset = \lambda$ and $\basek = d^2$.
The discussion on why this choice minimizes the overall complexity of the protocol appears in \Cref{sec:complexity}.

Here we discuss the selection of the field size. 
\begin{itemize}
    \item \textbf{In the base case ($t < \baset = \lambda$)}: 

    We invoke the \GKR protocol on $\CCbase$ with $\size(\CCbase) = O(t \cdot k \cdot S)$, $\depth(\CCbase) = O(\lambda + \log k + \log S)$,
    and thus $\depth(\CCbase) \cdot \log \size(\CCbase) = \lambda (\log (\lambda k S))^2$.
    Therefore, to make the soundness error less than $\GKRbound$, we set
    \[
    \abs{\FF} > 2^{\secloc}\cdot C_\GKR \lambda \cdot (\log (\lambda (k+1) S))^2
    \in O(2^{\secloc} T^3 k^2 S).
    \]
    \item \textbf{In the general case ($t \ge \baset$, $k \ge \basek$)}:

    For some absolute constant $c_2$,
    the following bound follows from \Cref{clm:description-LDP} in \Cref{sec:complexity},
    \begin{align*}
        \depth(\CCreduce) \log \size(\CCreduce) &\le \tO(\depth(C_{t, k}) \log \size(C_{t, k})) \\
        &= \max(\depth(C) \log \size(C),\lambda)\log(\abs{\FF}n S T)^{c_2},
    \end{align*}
    where $C_{t, k}$ is the implementation circuit of $\C_{t, k}$, 
    the input predicate defining $\cL_t^k[\C_{t, k}]$ in the recursive call to $\Batch(t, k)$,
    and $\CCreduce$ is the implementation circuit for the predicate on which we run \IPP.
    
    Let $\Gamma \coloneqq \depth(C) \log \size(C)+1$.
    In order to apply the soundness guarantee in \Cref{thm:IPP},
    the field size must be at least
    (recalling that $C_\GKR$ and $c$ are some absolute constants),
    \begin{align*}
    &C_{\GKR} \cdot 2^{\secloc + 4c + 7} \cdot ((\secloc d S(\lambda + 1)\log(\abs{\FF}(k+1)))^c + \depth(\CCreduce) \log \size(\CCreduce)+1)\\
    \in& O(2^{2\secloc} (\lambda^{c + 1}S^c \log k^c + \max(\Gamma,\lambda) \log(\abs{\FF} n S T k)^{c_2}))\\
    \subset& O(2^{4\secloc} T^5 k^4 \Gamma\cdot(S\Flog)^{\max(c, c_2) + 1}).
    \end{align*}
    Note that we used the actual \IPP invocation parameter $\secloc$, $d=\dbound$, and $(\lambda\log (nT(k+1)))^{O(1)} = o(Tk)$ to simplify the bounds.
\end{itemize}
To summarize, there exist some constants $c_0, c_2$ such that 
picking 
\[\abs{\FF} \ge \Fboundstmt\] 
would satisfy both requirements.\footnote{Specifically, 
$c_0$ is the hidden constant in the $O$ notation in the general case above and $c_1 = 2(\max(c, c_2) + 1)$.
Note that we also used the fact that $\abs{\FF}^{1/2} \gg \Fbits$ when $\abs{\FF}$ is large.}

\subsubsection{Soundness.}
\label{sec:main-soundness}
Suppose $\bm x \notin \cL_{t}^k[\C]$ and let $\cP^*$ be a cheating prover.
Let $\lrag \Psi \coloneqq \outputprotUU[\FF, \lrag{\C}, n, t, k, T]{}{^*(\bm x)}{}$ be the output of either $\Batch(t, k)$ if $k > \basek$ or $\SmallBatch(t, k)$ if $k \le \basek$,
and $E$ be the event that $\Psi(\bm x) = 1$.
Let $\epstk$ be the soundness error of $\Batch(t, k)$.
Our goal is to give an upper bound for $\Pr[E] \le \epstk$.

\begin{claim}
If $t < \baset$, $\eps < \GKRbound$.
\end{claim}
\begin{proof}
    In this case, the verifier simply returns the output of the \GKR protocol for checking $\Cbase(\bm x)$,
    which verifies that $\bm x \in \cL_t^k$.
    By the soundness of \GKR (\Cref{lem:GKR}), $\eps < \GKRbound$.
\end{proof}
Therefore, we consider the case when $t \ge \baset$.
Let $\bm M$ denote the midpoint matrix of $\bm x$.
We have the following cases depending on whether the batch size $k$ is small:
\begin{itemize}
    \item In protocol \SmallBatch, some matrix $\bm M^*$ is sent explicitly to the verifier.
    \item In protocol \Batch,
    let $\chi^*$ be the checksum sent by $\cP^*$.
    Since the checksum function is unique-decoding up to $d$ deviations,
    $\chi^*$ uniquely defines at most one matrix $\bm M^* \in \rowball(\bm M)$.
\end{itemize}
In either case, we can define the matrix $\bm M^* \in \bin^{k \times S(\lambda +1)}$. If such a matrix exists, define it uniquely; otherwise, set $\bm M^*=\bm M$ arbitrarily. Let $\xmid^*$ be the instance of $\cL_{\tmid}^{\kmid}$ read from $\bm M^*$. In what follows, if an event $A$ is characterized by a logical statement, we may use $A$ for both the event and the logical statement.

\begin{fact}
    Suppose $\bm x \notin \cL_{t}^k[\C]$, and
    define $F_1$, $F_2$, and $F_3$ as follows.
    \begin{enumerate}[label=($F_{\arabic*}$),ref=F_{\arabic*}]
        \item $\C(\bm M^*) = 1$ \label{eq:CRR}
        \item $\bm M^*[:, :S] = \bm x[:,:S]$ and $\bm M^*[:, -S:] = \bm x[:,-S:]$, \ie the leftmost start states and rightmost end states of $\bm M^*$ are the same as those in $\bm x$ \label{eq:boundary}
        \item $\xmid^* \in \cL_{\tmid}^{\kmid}$. \label{eq:midpoint}
    \end{enumerate}
    Then $\Pr[F_1 \cap  F_2 \cap F_3] = 0$, \ie $\Pr[\neg F_1 \cup \neg F_2 \cup \neg F_3] = 1$.
    \label{clm:F}
\end{fact}
\begin{proof}
    It suffices to show that $\ref{eq:boundary}\land\ref{eq:midpoint}\implies \neg \ref{eq:CRR}$. We have that $\ref{eq:boundary}\land\ref{eq:midpoint}$ implies $\bm M = \bm M^*$ and $\xmid = \xmid^*$ because $\cM$ is deterministic: (a) fixing the start states and (b) ensuring the transitions are correct would uniquely determine the entire path. 
    On the other hand, 
    $\bm x \notin \cL_{t}^k[\C]\Leftrightarrow \xmid \notin \cL_\tmid^\kmid \lor \C(\bm M) = 0$.
    Since $\xmid \in \cL_\tmid^\kmid$ (by $\ref{eq:midpoint}$ and the fact that $\xmid = \xmid^*$), we must have $\C(\bm M) = 0$.
\end{proof}
Therefore, by a union bound:
\begin{align*}
    \Pr[E] &\le \Pr[E \cap \neg F_1] + \Pr[E \cap \neg F_2] + \Pr[E \cap \neg F_3].
\end{align*}
    Let $G$ be the event that $\Cmid(\xmid^*) = 1$.
    By the inductive hypothesis, 
    the protocol $\Batch(\tmid, \kmid)$ at \Cref{ln:call-exp} in both versions of the protocol is sound,
    so
    \begin{fact}
        In both versions of the protocol,
        \Batch and \SmallBatch,
        $\Pr[G \mid \neg F_3] = \Pr[\Cmid(\xmid^*) = 1 \mid \xmid^* \notin \cLmid] \le \epsmid$.
        \label{eq:mid}
    \end{fact}
Moreover, by a union bound,
\begin{align*}
    \Pr[E] & \le \Pr[E \cap (\neg F_1 \cup \neg F_2 \cup \neg F_3)]\\
    &\le \Pr[E \cap \neg F_1] + \Pr[E \cap \neg F_2] + \Pr[E \cap \neg F_3].  \numberthis\label{eq:key}
\end{align*}
Now we can analyze the two versions of the protocol.
\begin{claim}
    In \SmallBatch, $\Pr[E] \le \epsmid$.
\end{claim}
\begin{proof}
    The verifier has access to the entire matrix $\bm M^*$,
    and the output $\Psi$ is the circuit $\Ccontract$ in \Cref{clm:small-batch},
    which checks $F_2 \land G$.
    \begin{itemize}
        \item Since the verifier checks that $\C(\bm M^*) = 1$ (\ie $\ref{eq:CRR}$) in \Cref{ln:CRR-basek} in \SmallBatch, if $\neg \ref{eq:CRR}$ holds,
        it rejects on that line,
        and thus $\Pr[E \mid \neg F_1] = 0$.
        \item Since the output circuit $\Psi$ produced by \SmallBatch checks $\bm M^*[:, 0:S] = \bm x[:, 0:S]$ (\ie $\ref{eq:boundary}$),
        and also $\Cmid(\xmid) = 1$ (\ie $G$),
        $\Pr[E \mid \neg F_2] = \Pr[E \mid \neg G] = 0$.
    \end{itemize}
    Therefore, by \Cref{eq:key,eq:mid},
    \begin{align*}
        \Pr[E] &\le 0 + \Pr[E \land \neg F_2] + \Pr[E \wedge \neg F_3]\\
        &\le 0 + \Pr[E \mid \neg F_2] + \Pr[E \land \neg F_3 \land \neg G] + \Pr[E \land \neg F_3 \land G]\\
        &\le 2 \cdot 0 + \Pr[E \mid \neg G] + Pr[G \mid \neg F_3] \le \epsmid.\qedhere
    \end{align*}
\end{proof}

\begin{claim}
    In \Batch,
    $\Pr[E] \le \epsred + \epsmid + \IPPbound$.
\end{claim}
\begin{proof}
    We first prove a series of facts.
    \begin{fact}
        In \Batch,
        $\Pr[E \mid \neg F_2] = 0$.
        \label{clm:7}
    \end{fact}
    \begin{proof}[Proof of \Cref{clm:7}]
        Recall that $\neg \ref{eq:boundary}$ is the event that $\bm M^*[:, 0:S] \neq \bm x[:, 0:S]$ or $\bm M^*[:, -S:] \neq \bm x[:, -S:]$. If $\bm M^*[:, 0:S] \neq \bm x[:, 0:S]$, by unique decoding (\Cref{lem:unique-decoding}), 
        $\cksum(\bm x[:, 0:S]) \neq \chi^*[:, 0:S] = \cksum(\bm M^*[:, 0:S])$.
        The output circuit $\Psi$ produced by \Batch checks that $\cksum(\bm x[:, 0:S]) = \chi^*[:, 0:S]$ (\Cref{clm:batch}),
        so indeed $\Pr[E \mid \neg F_2] = 0$. The same argument applies for $\bm M^*[:, -S:] \neq \bm x[:, -S:]$.
    \end{proof}
    Let $\kred \coloneqq \abs{\cQ}$,
    $\Mred \coloneqq \bm M[\cQ, :]$ be the corresponding variables after the protocol runs \Cref{ln:RR}.
    \begin{fact}
        $\Pr[\CIPP(\Mred) = 1 \mid \neg (F_1 \cap G)] \le 2^{-\secloc}$.
        \label{clm:8}
    \end{fact}
    \begin{proof}[Proof of \Cref{clm:8}]
        By the unique decoding property (\Cref{lem:unique-decoding}),
        $\rowball(\bm M) \cap \Creduce^{-1}(1) \subset \set{\bm M^*}$
        because $\Creduce$ checks that its input has checksum $\chi^*$.
        However, $\Creduce$ also checks $\ref{eq:CRR}$ and $G$,
        so $\bm M^*$ is excluded from its acceptance range when $\neg (\ref{eq:CRR} \land G)$,
        and $\bm M$ is $d$-$\Deltac$-far from $\cL_\Creduce$.

        The prerequisite of the \IPP soundness is also satisfied by our choice of $d$ and $\FF$,
        so \Cref{thm:IPP} entails that
        $\Pr[\CIPP(\Mred) = 1 \mid \neg (F_1 \cap G)] \le \IPPbound$.
    \end{proof}
    Let $\lrag \Cred$ be the output of \Cref{ln:red} in \Batch,
    and $\xred \coloneqq \bm x[\cQ, :]$.
    Applying the inductive hypothesis to the protocol $\Batch(t, \kred)$, we have
    \begin{align*}
        \Pr[E \mid \CIPP(\Mred) = 0] = \Pr[\Cred(\bm \xred) = 1 \mid \Cred(\Mred) = 0] \le \epsred.
    \end{align*}
    This, along with \Cref{clm:8}, gives
    \begin{align*}
        \Pr[E \mid \neg (F_1 \cap G)] &\le \Pr[E \mid \CIPP(\Mred) = 0] + \Pr[\CIPP(\Mred) = 1 \mid \neg (F_1 \cap G)]\\
        &\le \epsred + \IPPbound.
    \end{align*}
    Finally, combining with \Cref{clm:7,eq:mid,eq:key},
    \begin{align*}
        \Pr[E] &\le\Pr[E \cap (\neg F_1 \cup \neg F_3)] + \Pr[E \cap \neg F_2]\\
        &\le \Pr[E \cap (\neg F_3 \cap G)] + \Pr[E \cap (\neg F_1 \cup (\neg F_3 \cap \neg G))] + \Pr[E \mid \neg F_2]\\
        &\le \Pr[G \mid \neg F_3] + \Pr[E \mid \neg F_1 \cup \neg G] + 0\\
        &\le \epsred + \epsmid + \IPPbound. \qedhere
    \end{align*}
\end{proof}
To sum up the three cases,
\begin{align*}
    \Pr[E] \le \begin{cases} 
        \GKRbound & \text{if } t < \baset, \\
        \epsmid & \text{if } k < \basek, \\
        \epsred + \epsmid + \IPPbound.& \text{otherwise.}
    \end{cases}
\end{align*}

It remains to plug in our choice of parameters.
Let $\mathcal T(t,k)$ be the recursion tree generated by the two recursive calls.
By \Cref{sec:complexity}, $\abs{\mathcal T(t,k)} \le \recnodes$.
Every base node contributes at most $2^{-\secloc}$,
and every general node contributes at most $2^{-\secloc}$ in addition to its recursive errors.
Thus the inductive recurrence above gives
\[
\eps(t,k)
\le 2\abs{\mathcal T(t,k)}\cdot 2^{-\secloc}
\le 2\recnodes\cdot 2^{-\secpar-\ceil{\log(4\recnodes)}-2}
<2^{-\secpar}.
\]
\subsubsection{Complexities.}
\label{sec:complexity}
Let $\lrag{\Psi}$ be output by the protocol,
and $C_\Psi$ be the predicate's implementation circuit.
Let $\ell(t, k)$, $a(t, k)$, $\Ptime(t, k)$, $\Vtime(t, k)$ 
be the round complexity, per-round communication complexity, prover runtime, and verifier time of $\Batch(t, k)$,
and
$\size(t, k)$,
$\depth(t, k)$,
$\abs{\lrag \Psi}(t, k)$ be the size, depth and description length of the output predicate of $\Batch(t,k)$.

$\tO$ hides $\polylog(\abs{\FF}, n, S, T, k)$ factors.
Recall that $\tmid = t/\lambda$, $\kmid = k \cdot \lambda$, $\kred = \abs{\cQ} \le \ceil{24\secloc \frac{k}{d}}$.
\begin{claim}
The size, depth, and description length of the output low-depth predicate satisfy:
\begin{align*}
    \size(t, k) &= \begin{cases}
    \tO(k \cdot S) & \text{if } t < \baset, \\
    \size(\tmid, \kmid) + \tO(\basek S \lambda) & \text{if } k < \basek,\\
    \size(t, \kred) + \tO(k\cdot \poly(d) S) & \text{otherwise},
\end{cases} \\
    \depth(t, k) &= \begin{cases}
    \tO(\baset) & \text{if } t < \baset, \\
    \depth(\tmid, \kmid) + \tO(1) & \text{if } k < \basek,\\
    \depth(t, \kred) + \tO(1) & \text{otherwise},
\end{cases} \\
    \abs{\lrag{\Psi}}(t, k) &= \begin{cases}
    \tO(1) & \text{if } t < \baset, \\
    \abs{\lrag{\Psi}}(\tmid, \kmid) + \tO(\basek \cdot S\lambda) & \text{if } k < \basek,\\
    \abs{\lrag{\Psi}}(t, \kred) + \tO(dS + \poly(d)) & \text{otherwise},
\end{cases}
\end{align*}
Note that $\baset = \lambda$ and $\basek = d^2$.
Along any recursion path, at most $\log_\lambda T$ expanded-batch calls occur, so every encountered state satisfies $k'\le kT$.
Consequently $\depth(t', k') < \tO(\lambda)$, $\log\size(t',k') \le \polylog(k,T,S,\lambda,d)$, and $\abs{\lrag{\Psi}}(t', k') < \tO(d^2S\lambda + \poly(d))$ for all $(t', k')$ encountered in the recursion.
\label{clm:description-LDP}
\end{claim}
We analyze the complexities (and thereby prove \Cref{clm:description-LDP}) by the following cases.
\begin{itemize}
    \item 
    When $t < \baset$,
    by \Cref{lem:GKR,clm:base-case},
    we have
    \begin{itemize}
        \item $\size(\CCbase) = O(\baset \cdot k \cdot S)$,
        \item $\depth(\CCbase) = O(\baset + \log k + \log S) = \tO(\baset)$,
        \item $\ell(t, k) = \ell_\GKR(\CCbase) = \depth(\CCbase)\log \size(\CCbase) = \tO(\baset)$,
        \item $a(t, k) = a_\GKR(\CCbase) = O(\Flog) = \tO(1)$.
        \item $\Ptime(t, k) = \Ptime_\GKR(\CCbase) = \tO(\poly(\size(\CCbase) \cdot \Fbits)) = \poly(\baset kS)$,
        \item $\Vtime(t, k) = \Vtime_\GKR(\CCbase) = \tO((\ell_\GKR(\CCbase) + \abs{\lrag{\Cbase}}) \Flog) = \tO(\baset)$.
        \item $\size(t, k) = \tO(k \cdot S)$.
        \item $\depth(t, k) = \tO(1)$.
        \item $\abs{\lrag \Psi} = O(\Flog) = \tO(1)$.
    \end{itemize}

    \item When $k < \basek$, 
    we run the $\SmallBatch$ protocol, in which case
    \begin{itemize}
        \item $\begin{aligned}[t]
            \ell(t, k) &= \ell(\tmid, \kmid) + O(1).
        \end{aligned}$
        \item $\begin{aligned}[t]
            a(t, k) &= a(\tmid, \kmid) + O(\basek S \lambda).
        \end{aligned}$
        \item $\begin{aligned}[t]
            \Ptime(t, k) &= \Ptime(\tmid, \kmid) + \tO(\basek S \lambda).
        \end{aligned}$
        \item $\begin{aligned}[t]
            \Vtime(t, k) &= \Vtime(\tmid, \kmid) + \tO(\basek S \lambda) + \size(C) \cdot \Fbits.
        \end{aligned}$
    \end{itemize}
    The additional terms in the prover and verifier time are due to the time to send and receive the matrix $\bm M$ and the time to verify the predicate $\C$ by simulating its implementation circuit $C$ (with a $\Fbits$ overhead).
    Also note that $\Ptime$ does not pay any cost in computing the intermediate states at \Cref{ln:tableau} because we assume it has access to the entire tableau as auxiliary input.
    
    Furthermore, by \Cref{clm:small-batch},
    \begin{itemize}
        \item $\size(t, k) = \size(\CCSmallBatch) = \tO(\basek S \lambda) + \size(\tmid, \kmid)$.
        \item $\depth(t, k) = \depth(\CCSmallBatch) = \tO(1) + \depth(\tmid, \kmid)$.
        \item $\abs{\lrag \Psi}(t, k) = \tO(\basek S \lambda) + \abs{\lrag \Cmid} = \tO(\basek S \lambda) + \abs{\lrag \C}(\tmid, \kmid)$.
    \end{itemize}
 
    \item 
    In the general case,
    let $\tmid = t/\lambda$, $\kmid = k\lambda$, and $\kred = \abs{\cQ}$.
    By \Cref{thm:IPP}, $\abs{\cQ} \le \ceil{24\secloc\cdot\frac{k}{d}} < \frac{k}{\lambda}$,
    where the final inequality follows from $d=\dbound$.
    We first work out the bounds for intermediate parameters, using \Cref{clm:reduction-case} and \Cref{clm:description-LDP}.
    \begin{itemize}
        \item $\cksumT = 2d$.
        \item $\size(\CCreduce) = \tO(k\cksumT S\lambda) + \size(C) + \size(\CCmid)$, and therefore $\log\size(\CCreduce)\le \polylog(k,T,S,\lambda,d,\abs{\FF})+\log\size(C)$.
        \item $\depth(\CCreduce) = \tO(\max(\depth(C), \depth(\CCmid))) = \tO(\depth(C))$.
    \end{itemize}
    Therefore,
    $\depth(\CCreduce) \log \size(\CCreduce) = \tO(\depth(C) \log \size (C)+\lambda)$.
    By \Cref{thm:IPP},
    \begin{itemize}
        \item $\begin{aligned}[t]
            \ell(t, k) &= \ell(\tmid, \kmid) + \ell(t, \kred) + \ell_\IPP(\CCreduce) + O(1) \\
            &= \ell(\tmid, \kmid) + \ell(t, \kred) + \tO(\depth(\CCreduce) \log \size(\CCreduce))\\
            &= \ell(\tmid, \kmid) + \ell(t, \kred) + \tO(\depth(C) \log \size(C)+\lambda).
        \end{aligned}$
        \item $\begin{aligned}[t]
            a(t, k) &= a(\tmid, \kmid) + a(t, \kred) + O(\cksumT S \lambda \Flog) + a_\IPP(\CCreduce)\\
            &= a(\tmid, \kmid) + a(t, \kred) + \tO(S\lambda\poly(d)).
        \end{aligned}$
        \item $\begin{aligned}[t]
            \Ptime(t, k) &= \Ptime(\tmid, \kmid) + \Ptime(t, \kred) + \tO(k \cdot \cksumT \cdot S \cdot \lambda \cdot \Fbits)\\
            &\quad + \Ptime_\IPP(\CCreduce)\\
            &= \Ptime(\tmid, \kmid) + \Ptime(t, \kred) + \tO(\poly(kdS\lambda\size(C))).
        \end{aligned}$
        \item $\begin{aligned}[t]
            \Vtime(t, k) &= \Vtime(\tmid, \kmid) + \Vtime(t, \kred)+ \Vtime_\IPP(\CCreduce)\\
            &= \Vtime(\tmid, \kmid) + \Vtime(t, \kred) \\
            &\quad+ \tO((\depth(\CCreduce) \log \size(\CCreduce) + \abs{\lrag{\Creduce}}) \cdot Sd\lambda + \poly(d))\\
            &= \Vtime(\tmid, \kmid) + \Vtime(t, \kred) \\
            &\quad+ \tO((\depth(C)\log \size(C) + \lambda + \abs{\lrag \C} + \abs{\lrag{\Cmid}} + dS\lambda) Sd \lambda + \poly(d))\\
            &= \Vtime(\tmid, \kmid) + \Vtime(t, \kred) \\
            &\quad+ \tO((\depth(C) \log \size(C) + \abs{\lrag{\C}})S^2\lambda^2\poly(d)).
        \end{aligned}$\\
        Note that we used \Cref{clm:description-LDP}, \ie that $\abs{\lrag{\Cmid}} \le d^2 S\lambda + \poly(d)$,
        and that $\basek = \tO(d^2)$ in deriving the upper bound on $\Vtime$.
    \end{itemize}
    By \Cref{clm:batch},
    \begin{itemize}
        \item $\size(t, k) = \size(\CCBatch) = \tO(k\cksumT S + \abs{\cQ}(\size(G) + S) + \size(\CCred)) = \tO(kdS + k(\poly(d) + S) + \size(t, \kred)) = \tO(k\cdot \poly(d) S + \size(t, \kred))$.
        \item $\depth(t, k) = \depth(\CCBatch) = \tO(1) + \tO(G) + \depth(\CCred) = \tO(1) + \depth(t, \kred)$.
        \item $\abs{\lrag{\Psi}}(t, k) = \abs{\lrag{\CBatch}} = \abs{\lrag{\Cred}} + \abs{\lrag \cQ} + \tO(dS) = \abs{\lrag{\Psi}}(t, \kred) + \tO(dS + \poly(d))$.
    \end{itemize}
\end{itemize}
To summarize,
\begin{align*}
    \size(t, k) &= \begin{cases}
    \tO(k \cdot S) & \text{if } t < \baset, \\
    \size(\tmid, \kmid) + \tO(\basek S \lambda) & \text{if } k < \basek,\\
    \size(t, \kred) + \tO(k\cdot \poly(d) S) & \text{otherwise}.
\end{cases}
\end{align*}
\begin{align*}
    \depth(t, k) &= \begin{cases}
    \tO(\baset) & \text{if } t < \baset, \\
    \depth(\tmid, \kmid) + \tO(1) & \text{if } k < \basek,\\
    \depth(t, \kred) + \tO(1) & \text{otherwise}.
\end{cases}
\end{align*}
\begin{align*}
    \ell(t, k) &= \begin{cases}
    \tO(\baset) & \text{if } t < \baset, \\
    \ell(\tmid, \kmid) + \tO(1) & \text{if } k < \basek,\\
    \ell(\tmid, \kmid) + \ell(t, \kred) + \tO(\depth(C) \log \size(C)) & \text{otherwise},
\end{cases}
\end{align*}
\begin{align*}
a(t, k) &= \begin{cases}
    \tO(1) & \text{if } t < \baset, \\
    a(\tmid, \kmid) + \tO(\basek \cdot S\lambda) & \text{if } k < \basek,\\
    a(\tmid, \kmid) + a(t, \kred) + \tO(S\lambda\poly(d)) & \text{otherwise},
\end{cases}
\end{align*}
\begin{align*}
\Ptime(t, k) &= \begin{cases}
    \poly(\baset kS) & \text{if } t < \baset,\\
    \Ptime(\tmid, \kmid) + \tO(\poly(\basek S\lambda)) & \text{if } k < \basek,\\
    \Ptime(\tmid, \kmid) + \Ptime(t, \kred) + \tO(\poly(kdS\lambda\size(C))) & \text{otherwise},
\end{cases}
\end{align*}
\begin{align*}
\Vtime(t, k) &= \begin{cases}
    \tO(\baset) & \text{if } t < \baset, \\
    \Vtime(\tmid, \kmid) + \tO(\basek \cdot S\lambda) & \text{if } k < \basek,\\
    \Vtime(\tmid, \kmid) + \Vtime(t, \kred) \\
    \quad + (\depth(C) \log \size(C) + \abs{\lrag{\C}})S^2\lambda^2\poly(d)) & \text{otherwise}.
\end{cases}
\end{align*}
We set the parameter $\lambda = \Lambda(T, k) = 2^{\sqrt{\log \binom{2\log T + \log k}{\log T}}}$.
As analyzed below, the number of nodes in the recurrence tree is bounded by $\binom{2\log_\lambda T + \log_\lambda k}{\log_\lambda T}$, which is upper bounded by $\poly(\lambda)$ for our choice of $\lambda$.
Solving the recurrence relations,
plugging in $\baset = \lambda$, $\basek = d^2$, $d = \poly(\lambda,\secpar)$,
and bounding the size of the recurrence tree by $\poly(\lambda)$,
we have
\begin{align*}
    \ell(t, k) &\le \poly(\lambda) \cdot \tO(\lambda) + \tO(\depth(C) \log \size(C)) = (\poly(\lambda) + \depth(C) \log \size(C)) \cdot \poly(\secpar) .\\
    a(t, k) &\le \poly(\lambda) \cdot \poly(\lambda) \cdot \tO(S) = \poly(\lambda) \cdot S \cdot \poly(\secpar).\\ 
    \Ptime(t, k) &\le \poly(\lambda) \cdot \poly(T,S,\lambda,\size(C)) = \poly(S, T, \secpar,\size(C)).\\
    \Vtime(t, k) &\le (\poly(\lambda) + \depth(C) \log \size(C) + \abs{\lrag{\C}}) \cdot \poly(\lambda)\cdot \tO(S^2)\\
    &= (\poly(\lambda) + \depth(C) \log \size(C) + \abs{\lrag{\C}}) \cdot S^2 \cdot \poly(\secpar).\\
    \size(t, k) &\le (k+\poly(d)\lambda)\cdot S\cdot\poly(d) = (k+\poly(\lambda)) \cdot S \cdot \poly(\secpar).\\
    \depth(t, k) &\le \tO(\lambda) = \poly(\lambda) \cdot \poly(\secpar), \\
    \abs{\lrag{\Psi}}(t, k) &\le \tO(d^2S\lambda + \poly(d)) = \poly(\lambda) \cdot S \cdot \poly(\secpar),
\end{align*}
Note that $\tO$ hides $\polylog(\abs{\FF},n,S,T,k) = \poly(\secpar)\cdot \polylog(n,S,T,k)$ terms,
and we simplify the bounds for $\ell$ and $\Vtime$ by \Cref{clm:description-LDP},
which implies that
$\depth(t', k')\log \size(t', k') \le \tO(\lambda)$ and $\abs{\lrag \C}(t',k') \le \tO(\poly(\lambda) \cdot S)$ for all $t', k'$ that appear in the recurrence.

We analyze our selection of parameters $d, \lambda, \baset, \basek$ and show that they yield the best complexity in \Cref{sec:parameter-optimality}.

\subsection{Construction of Auxiliary Circuits}
\label{sec:auxiliary-claims}
\begin{proof}[Proof of \Cref{clm:base-case}]
  We construct a circuit $C_\mathsf{base}$ with the given bounds that
  checks that $\bm x\in\mathcal{L}_t^{k}$. 
 To do this, we can use the well-known result that
  any time-bounded Turing machine can be efficiently simulated by a
  circuit. Observe that there is a circuit of size $S$ and
  constant depth (in the size of the alphabet and the number of states
  of $\mathcal{M}$) that, given a configuration of $\mathcal{M}$,
  outputs the configuration at the next time step. Call this circuit $C_1$. From $C_1$, we can
  construct a circuit $C_t$ that, given a starting configuration,
  outputs the configuration of $\mathcal{M}$ after $t$ time steps, by
  layering $t$ copies of $C_1$ on top of each other. Then we can
  construct a circuit that checks whether $\mathcal{M}$ goes from some
  starting configuration to some ending configuration in $t$ time
  steps by appending onto $C_t$ a circuit that checks equality between
  the output of $C_t$ and the ending configuration, which can be done using $S$ gates and $\log S$ depth. We can do this for $k$ $t$-time computations in parallel at the cost of a multiplicative factor of $k$ in the size and an additive factor of $\log k$ in the depth.

  So putting everything together, we have $\mathsf{size}(C_\mathsf{base}) = O(t\cdot k\cdot S)$ and
  $\mathsf{depth}(C_\mathsf{base}) = O(t+\log k+\log S)$.

\end{proof}

\begin{proof}[Proof of \Cref{clm:batch}]
We construct a circuit $\ensuremath{C_{\Batch}}\xspace$ that satisfies the given bounds. The circuit $\ensuremath{C_{\Batch}}\xspace$ needs to (1) check that $\mathsf{cksum}_d({\bm x}) = \cksumv_{\mathsf{bdry}}$, (2) expand $\langle\cQ\rangle$ and compute $\xred = \bm x[\cQ,:]$, and (3) verify that $\Cred(\xred) = 1$. To check (1), we compute the checksum of every one of the $2S$ columns of ${\bm x}$. By \Cref{lem:unique-decoding}, computing the checksum of one column can be done with a circuit of size $O(k\cksumT)$ and depth $\tO(1)$, and here $\cksumT=2d$. Therefore the $2S$ boundary-column checksums can be computed using $\tO(k\cksumT S)$ gates and $\tO(1)$ depth. Then we verify equality with $\cksumv_{\mathsf{bdry}}$, which can be done using $\tO(\cksumT S)$ gates and $\tO(1)$ depth. For condition (2), expanding $\langle\cQ\rangle$ can be done using $|\mathcal{Q}|\mathsf{size}(G)$ gates and $\mathsf{depth}(G)$ depth (by taking $|\mathcal{Q}|$ parallel copies of $G$), and then computing $\xred = \bm x[\cQ,:]$ can be done using $O(|\mathcal{Q}|S)$ gates and $\widetilde{O}(1)$ depth. For condition (3), we call $\CCred$ on $(\xred, \langle\Cred\rangle)$. So in total, we have $\size(\CCBatch) = \tO(k\cksumT S+ |\mathcal{Q}|(\mathsf{size}(G) + S) + \mathsf{size}(\CCred))$ and $\depth(\CCBatch) = \tO(1) + \mathsf{depth}(G) + \mathsf{depth}(\CCred)$. 
\end{proof}

\begin{proof}[Proof of \Cref{clm:reduction-case}]
We construct a circuit $C_\mathsf{reduce}$ with the given bounds
that checks the following conditions.
  \begin{enumerate}
  \item Reads out
    $\xmid = ((\xmid_{i,j}, \xmid_{i,j + 1}))_{i, j \in [k]                       
      \times [0, \lambda - 1]} \in \bin^{2 S(n) \cdot k \cdot                        
      \lambda}$ from $\bm M$ and verifies that
    $\Cmid(\xmid) = 1$.
    To do this, we call $\CCmid(\xmid, \lrag{\Cmid})$.
    This can be done using a circuit of size
    $\tO(kS\lambda) + \size(\CCmid)$ and depth $\tO(1) + \depth(\CCmid)$.
  \item Verifies that $\cksum_d(\bm M) = \cksumv$,
    where $\cksumv \in \FF^{\cksumT \times (S(\lambda + 1))}$ are the hard-coded checksums.
    To do this, we compute the checksum of every column of $\bm M$ and
    verify that the result matches $\cksumv$. By
    \Cref{lem:unique-decoding}, computing the checksum of
    one column requires a circuit of size $O(k\cksumT)$ and depth $\tO(1)$, with $\cksumT=2d$.
    Therefore, the checksums of all columns can be computed using a circuit of size
    $\tO(k\cksumT S\lambda)$ and depth $\tO(1)$. 
    Then we can check
    equality with $\cksumv$ using a circuit of size
    $\tO(\cksumT S\lambda)$ and
    depth $\tO(1)$. So the total complexity is
    $\tO(k\cksumT S\lambda)$ gates
    and depth $\tO(1)$.
  \item Verifies that $\C(\bm M) = 1$.
    We can do this by calling $C$ on $({\bm M}, \langle\C\rangle)$. 
  \end{enumerate}
So we have
\[
\mathsf{size}(C_\mathsf{reduce}) = \tO(k\cksumT S\lambda + \mathsf{size}(C) + \mathsf{size}(\CCmid))
\]
and 
\[
\mathsf{depth}(C_\mathsf{reduce}) =        
\tO(\max(\mathsf{depth}(C), \mathsf{depth}(\CCmid))).
\]
\end{proof}

\begin{proof}[Proof of \Cref{clm:small-batch}]

We construct a circuit $\CCcontract$ that satisfies the given bounds. The circuit $\CCcontract$ needs to verify that (1) $\xmid$'s boundary states agree with $\bm x$ and (2) $\Cmid(\xmid) = 1$. Condition (1) is an equality check, which can be done with a circuit with $O(kS)$ gates and $\tO(1)$ depth. To verify (2), we call $\CCmid(\xmid, \lrag{\Cmid})$. So in total, we have $\size(\CCSmallBatch) = \tO(kS\lambda) + \size(\CCmid)$ and $\depth(\CCSmallBatch) = \tO(1) + \depth(\CCmid)$. 

\end{proof}

\section{Acknowledgments}
Y.T.K. is supported by Amazon grant.
M.M.H., and Y.T.K. are supported by the Defense Advanced Research Projects Agency (DARPA) under Contract No. HR0011-25-C-0300 (to Y.T.K.). Z.X. is supported by an Akamai Presidential Fellowship.
Any opinions, findings and conclusions or recommendations expressed in this material are those of the author(s) and do not necessarily reflect the views of the Defense Advanced Research Projects Agency (DARPA).
\pagestyle{plain}

\bibliographystyle{alpha}
\bibliography{abbrev3,crypto,ref}

\appendix
\section{Interactive Proof of Proximity with Row Reduction}
\label{sec:IPP-proof}
We show the proof of \Cref{thm:IPP}, which we restate here for completeness.

\RRIPP*

\paragraph{Organization} We first set up additional preliminaries in \Cref{sec:LDE,sec:GKR-preserving,sec:RR-PVAL},
and we present the formal protocol and its analysis in \Cref{sec:IPP-analysis}.

\subsection{Low Degree Extension and Polynomial Valuation (\pval)}
\label{sec:LDE}
    

Given a subset $H \subset \FF$
and an integer $m \in \NN$,
the \emph{low-degree extension} (LDE) of a function $f : H^m \to \FF$ is the unique $(\abs{H} - 1)$-individual-degree polynomial $\hat {f}: \FF^{m} \to \FF$ such that for all $\bm s \in H^m$, $\hat {f}(\bm s) = f(\bm s)$.
In this work we focus on the special case when $H = \{0,1\}$,
where $\hat {f}$ is multilinear over $\FF^m$,
and is called the multilinear extension of $f$.
Abusing the notation, 
given a string $x \in \FF^n$,
we first pad it with 0's such that its length becomes $2^m$ where $m = \ceil{\log n}$,
and let $\hat{x}: \FF^m \to \FF$ denote the multilinear extension of the function $f_{x} : \{0,1\}^m \to \FF$ whose function table is $x$.
On a sequence $\bm \pj = (\bm \pj_1, \ldots, \bm \pj_\pvalT)\in (\FF^m)^{\pvalT}$ of length $\pvalT$,
we use the shorthand $\hat x(\bm \pj) \coloneqq (\hat x(\bm \pj_1),\ldots,\hat x(\bm \pj_\pvalT)) \in \FF^\pvalT$.

Consider the following \emph{Polynomial Valuation} (\pval) set, which is an affine subspace over $\FF$, first defined in~\cite{STOC:RotVadWig13}.

\begin{definition}[The \pval set]
\label{def:pval}
    Let $m, \pvalT \in \NN$.
    The set $\pvalF(\bm \pj, \bm \pv) \subset \FF^{2^m}$ is parameterized by the sequences 
    $\bm \pj = (\pj_1,\ldots,\pj_\pvalT) \in (\FF^m)^\pvalT$ and $\bm \pv = (\pv_1,\ldots,\pv_\pvalT) \in \FF^\pvalT$. 
    It consists of all strings $x \in \FF^{2^m}$ whose corresponding multilinear extension $\hat x: \FF^m \to \FF$ satisfies $\hat x(\bm \pj) = \bm \pv$.
\end{definition}
Note that $\pval(\bm \pj, \bm 0)$ is a linear subspace of $\FF^{2^m}$,
so we can define $\Delta(\pval(\bm \pj, \bm 0))$ to be the minimum Hamming distance of a non-zero vector in $\pval(\bm \pj, \bm 0)$.

\subsection{The \GKR protocol is $\Delta_c$-distance-preserving}
\label{sec:GKR-preserving}
The description $\lrag \Psi$ output by the \GKR protocol is a pair $(\pj, \pv) \in \FF^{\ceil{\log n}} \times \FF$,
such that $\Psi(x) = 1$ iff $\hat x(\pj) = \pv$.
Regardless of $\cP^*$, $\pj$ is always uniformly random and only depends on the random coins sampled by $\cV$.

Let $\prot{_\GKR}{}{}$ be the \GKR protocol from \Cref{lem:GKR}.
In \cite{STOC:RotVadWig13},
the authors observed that when $\prot{_\GKR}{}{}$ is parallel-repeated $\pvalT$ times for a large enough $\pvalT$,
it is distance-preserving in the following sense:
suppose the input $x$ is $d$-Hamming-far from satisfying $\C$,
and $\bm \pj = (\pj_1,\ldots,\pj_\pvalT)$ and $\bm \pv = (\pv_1,\ldots,\pv_\pvalT)$ are the outputs of $\pvalT$ parallel repetitions of $\prot[\FF, \lrag \C]{_\GKR}{(x)}{}$,
then $x$ is $d$-Hamming-far from the set $\pval(\bm \pj, \bm \pv)$.
Note that this increases the overall cost of the protocol by a factor of $\pvalT$.
This generalizes to $\Delta_c$-distance. 

In the following, we naturally interpret $\pval(\pj, \bm \pv)$ as a subspace of $\FF^{k \times L}$, by treating matrices in $\FF^{k \times L}$ as vectors in $\FF^{kL}$. $\Deltac(\pval(\pj, \bm 0))$ is the minimum $\Deltac$ distance of matrices in $\pval(\pj, \bm 0)$.

\begin{lemma}[\GKR is $\Delta_c$-Distance-Preserving; Summary of Lemma 2 and Claim 3 in \cite{FOCS:BGHK25}]
    \label{lem:RVW}
    Let $k, L, d\in \NN$.
    Denote by $\FF$ a field of characteristic 2.
    Consider a log-space uniform arithmetic circuit $\C: \FF^{k \times L} \to \bin$ over $\FF$,
    with addition and multiplication gates of fan-in 2.
    Let $C$ be its implementation circuit,
    and let $S$ and $D$ be its size and depth, respectively.
    Let $\lrag{\C}$ be the description of $\C$.

    Let $\cL_\C \coloneqq \set{\bm M \in \FF^{k \times L}: \C(\bm M) = 1}$ denote the set of strings accepted by $\C$. 
    Suppose $\bm M \in \FF^{k \times L}$,
    and that
    \[
    \cL_\C \cap \rowball(\bm M)  = \varnothing,
    \]
    \ie nothing in the ball $\rowball(\bm M)$ satisfies the circuit. (See \Cref{def:Deltac-dist} for the definition of $\rowball$.)

    There is a constant $C_\GKR > 0$ such that the following holds. For any prover $\cP^*$ and $\secpar \in \NN$,
    let $(\pj, \bm \pv)$ be the output $\outputprotUU[\FF, \lrag \C]{_{\GKR}^*(\bm M)}{_{\GKR}}$ with $\pvalT \ge 8dL\log k+ \secpar + 3$ parallel repetitions,
    \begin{align}
        \Pr[\pval(\pj, \bm \pv) \cap \rowball(\bm M) \neq \varnothing] &\le \left(\left(C_\GKR \cdot \frac{D \log S}{\abs{\FF}}\right)^\pvalT \cdot\left(\binom{k}{d} \abs{\FF}^d\right)^\ncol\right). \label{eq:RVW}\\
        \Pr[\Deltac(\pval(\pj, \bm 0)) < 4d] &\le 2^{-\secpar-2}. \label{eq:kernel}
    \end{align}
\end{lemma}

\subsection{The Interactive Proof of Proximity for \pval with Row Reduction}
\label{sec:RR-PVAL}
\begin{lemma}[The Row Reduction Protocol for \pval with $\Delta_c$-distance, Theorem~9 in \cite{FOCS:BGHK25}; c.f. \cite{TCC:RotRot20}]
    \label{lem:DcRR}
    Suppose $k, L, d \in \NN$.
    Let $\FF$ be a field.
    Suppose $\secpar, \pvalT \in \NN$.
    Let $\bm \pj = (\bm \pj_1, \ldots, \bm \pj_\pvalT) \in (\FF^{\log k+ \log L})^\pvalT$ and $\bm \pv = (v_1, \ldots, v_\pvalT) \in \FF^\pvalT$,
    and let $\bm M \in \FF^{k \times L}$.
    For some constant $c > 0$,
    if the following conditions hold:
    \begin{itemize}
        \item $d \ge \dboundm \in O(\secpar \log k)$,
        \item $\abs{\FF} \ge \FboundT \in O(2^\secpar \cdot \poly(R, \log k, \log L))$,
    \end{itemize}
    then there exists an $\protcplx{}$ protocol $\prot[\secpar, \FF, d, \pj, \bm \pv]{_\RR}{(\bm M)}{}$,
    which either outputs $\bot$ or $(\lrag \cQ, \lrag \CRR)$ that are descriptions of 
    a set of rows $\cQ \subsetneq [k]$,
    and a predicate $\CRR$ satisfying the following properties.
    \begin{itemize}[label=-]
        \item \textbf{Completeness:} If $\bm M \in \pval(\bm \pj, \bm \pv)$,
        then $\Pr[\CRR({\bm M}[\cQ, :]) = 1] = 1$.
        \item $2^{-\secpar}$-\textbf{Soundness:} 
        Suppose $\Delta_c(\pval(\bm \pj,\bm 0)) \geq 4d$,
        and $\rowball(\bm M) \cap \pval(\pj, \bm \pv) = \varnothing$,\footnote{Recall that $\rowball(\bm M)$ is the set of all matrices that are $\Deltac$-d-close to $\bm M$}
        then for any prover strategy $\cP^*$,
        \[
        \Pr\left[\outputprotUU[\secpar, \FF, d, \pj, \bm \pv]{_\RR^*(\bm M)}{_\RR} = (\lrag \cQ, \lrag \CRR) \text{ s.t. } \CRR({\bm M}[\cQ, :]) = 1\right] \le 2^{-\secpar}.
        \]
        \item \textbf{Row Reduction:} 
        The subset of rows $\cQ \subsetneq [k]$ has size $\abs{\cQ} \le \ceil{8\secpar \cdot \frac{k}{d}}$.
    \end{itemize}
    
    With $\tO$ hiding $\polylog(\abs{\FF}, k, L)$ factors,
    if $\pvalT = \tO(d L)$,
    then the complexity of the protocol is as follows.
    \begin{itemize}
        \item $\ell = \tO(1)$.
        \item $a = \tO(\aboundIPP)$.
        \item $\Ptime = \poly(k L, \pvalT \cdot \Flog)$.
        \item $\Vtime = \tO(\aboundIPP)$.
    \end{itemize}
    The bit-lengths are $\abs{\lrag{\cQ}} = \tO(\poly(d))$ and $\abs{\lrag{\CRR}} = \tO(\aboundIPP)$.
    Let $G$ and $C$ be the implementation circuits of $\cQ$ and $\CRR$, respectively.
    Then they satisfy the following.
    \begin{itemize}
        \item $\size(G) = \tO(\poly(d))$.
        \item $\depth(G) = \tO(1)$.
        \item $\size(C) = \tO(\abs{\cQ} \cdot L)$.
        \item $\depth(C) = \tO(1)$.
    \end{itemize}
\end{lemma}

\subsection{Proof of \Cref{thm:IPP}}
\label{sec:IPP-analysis}
We present the protocol in \Cref{alg:IPP}.
It simply runs the \GKR protocol followed by the \RR protocol with appropriate parameters.

\begin{algorithm}[ht]
  \setstretch{1.2}
  \caption{Protocol $\IPP = \prot[\secpar, \FF, \lrag{\C}, d]{}{(\bm M)}{}$ for checking $\C(\bm M) = 1$.}
  \label{alg:IPP}
    \textbf{Input Parameters:} $\secpar \in \NN$, $\FF$ is a field, $\lrag \C$ describes a predicate $\C : \bin^{k \times L} \to \bin$ whose implementation circuit has size $S$ and depth $D$, $d \in \NN$.\\
    \textbf{Input Requirement:} $d \ge 48 \secpar \log k$, $\abs{\FF} \ge \FboundIPP$.\\
    \textbf{Input Matrix:} $\bm M \in \bin^{k \times L}$.\\
    \textbf{Derived Parameters:} $\pvalT = \Tbound = \tO(dL)$.\\
    \textbf{Verifier Output:} $\lrag \cQ, \lrag \CRR$.
    \begin{enumerate}[label=(\arabic*)]
        \item Apply the \textbf{GKR} protocol to create the intermediate claim $\bm M \in \pval(\pj, \pv)$. 

        \begin{algorithmic}
           \State Run $\prot[\FF, \lrag \C]{_\GKR}{(\bm M)}{}$ $\pvalT$ times in parallel and obtain $\lrag{\Psi} = (\pjM, \pvM)$. \label{ln:GKR-reduce}
        \end{algorithmic}
        \item Apply the \textbf{RR} protocol.
        \begin{algorithmic}
            \State Run $\prot[\sigma+2,\FF, d, \pjM, \pvM]{_\Reduce}{(\bm M)}{}$ and obtain $\lrag \cQ, \lrag \CIPP$.\\
            \Return $\lrag \cQ, \lrag \CIPP$.
        \end{algorithmic}
    \end{enumerate}
\end{algorithm}

\textbf{Completeness} follows from the completeness of both sub-protocols.
Assuming $\C(\bm M) = 1$, we have:
\begin{itemize}
    \item The $\GKR$ protocol outputs $(\pj, \pv)$ such that $\widehat{\bm M}(\pj) = \bm \pv$.
    In other words, $\bm M \in \pval(\pj, \bm \pv)$.
    \item By the completeness of the $\RR$ protocol,
    given that $\bm M \in \pval(\pj, \bm \pv)$,
    we have $\Psi(\bm M[\cQ, :]) = 1$.
\end{itemize}

\paragraph{Soundness}
Let $\cS \coloneqq \rowball(\bm M) \cap \pval(\pj, \pv)$.
If $\cL_\C \cap \rowball(\bm M) = \varnothing$,
then
given our parameter settings,
the following holds.
\begin{align*}
    \pvalT &\ge \Tbound,\\
    \abs{\FF} &> 2^{\secpar + 2} \cdot (D \log S),
\end{align*}
so by \Cref{lem:RVW},
\begin{align}
    \Pr[\cS \neq \varnothing] \le \left(C_\GKR \cdot \frac{D\log S}{\abs{\FF}}\right)^\pvalT \cdot\left(\binom{k}{d} \abs{\FF}^d\right)^\ncol &\le \epsIPPconcrete. \\
    \Pr[\Deltac(\pval(\pj, \bm 0)) < 4d] &\le \epsIPPconcrete.
\end{align}
Moreover,
the prerequisite of \Cref{lem:DcRR},
$d \ge 48 \secpar \log k \ge 16 (\secpar + 2) \log k$ and $\abs{\FF} \ge \FboundT$,
is also satisfied by the choice of the parameters,
so
\begin{align*}
    \Pr[\Psi(\bm M[\cQ, :]) = 1  \mid \cL_\C \cap \rowball(\bm M) = \varnothing]
    &\le \Pr[\cS \neq \varnothing \mid \cL_\C \cap \rowball(\bm M) = \varnothing] \\
    &\quad +\Pr[\Psi(\bm M[\cQ, :]) = 1 \mid \cS = \varnothing]\\
    &\le \epsIPPconcrete \\
    &\quad + \Pr[\CIPP(\Mred) = 1 \mid \cS = \varnothing, \Deltac(\pval(\pjM, \bm 0)) \ge 4d]\label{eq:clm42}\numberthis\\
    &\quad + \Pr[\Deltac(\pval(\pj, \bm 0)) < 4d] \\
    &\le 3 \cdot \epsIPPconcrete < 2^{-\secpar} \numberthis\label{eq:clm43}
\end{align*}
\Cref{eq:clm42} follows from \Cref{lem:RVW},
and \Cref{eq:clm43} follows from soundness (\Cref{lem:DcRR}) of the \RR protocol as well as \Cref{eq:kernel} in \Cref{lem:RVW}.

\paragraph{Complexities}
The overall cost is the sum of $R$ runs of the \GKR protocol and one run of the \RR protocol.
With $\tO$ hiding $\polylog(\abs{\FF}, k, L)$ factors,
and letting $S = \size(C)$, $D = \depth(C)$,
we have
\begin{itemize}
    \item $\ell = \tO(D\log S) + \tO(1) = \tO(D \log S)$.

    Note that the round complexity of $R$ \GKR protocol runs is the same as that of one run since they are run in parallel.
    \item $a = \tO(\max(\pvalT \cdot \Flog, \aboundIPP)) = \tO(d\aboundIPP)$,
    given that 
    \begin{align*}
    \pvalT &= \Tbound = \tO(dL + \secpar) = \tO(dL)\\
    \abs{\FF} &= \FboundIPP.
    \end{align*}
    \item $\Ptime = \poly(k, L, d, S,\Flog)$.
    \item $\Vtime = \tO(\pvalT \cdot (D\log S + \abs{\lrag \C}) \cdot \Flog + \aboundIPP) = \tO(dL(D \log S + \abs{\lrag \C}) + \poly(d))$.
\end{itemize}

\section{Analyzing Protocol Parameters}
\label{sec:parameter-optimality}

In this section, we show that our selection of parameters $d,\lambda,\baset,\basek$ for \Cref{thm:main-LDP} is optimal. 

To select the optimal parameters,
we focus on the recurrence relation on $\Vtime(t, k)$.
We also assume $\depth(C)\log\size(C) = \abs{\lrag{\C}} = O(1)$,
since we are most interested in $\cL_T = \cL_T^1[\top]$,
where $\top$ is the dummy predicate that accepts all strings. 
Since the base cases and the intermediate terms all have to be efficient,
we set all four parameters to be $\poly(n)$.
Let $v(p, q) \coloneqq \Vtime(\lambda^a, \lambda^b)$,
$\basep \coloneqq \log_\lambda \baset$, $\baseq \coloneqq \log_\lambda \basek$.
Then
\begin{align*}
    v(p, q) &= \begin{cases}
        \tO(\lambda^\basep) & \text{if } p < \basep, \\
        v(p-1, q+1) + \tO(\lambda^{\baseq + 1} \cdot S) & \text{if } q < \baseq,\\
        v(p-1, q+1) + v(p, q-\log_\lambda \frac{d}{24\secpar}) + \tO(S^2\lambda^2\poly(d)) & \text{otherwise.}
    \end{cases}
\end{align*}
Let $\gamma \coloneqq \frac{1}{\log_\lambda \frac{d}{24\secpar}}$.
We are interested in the quantity $\Vtime(T, k) = v(\log_\lambda T, \log_\lambda k)$.
Let us call the recurrence going from $v(p, q)$ to $v(p-1, q+1)$ a $\nwarrow$-step and the recurrence going from $v(p, q)$ to $v(p, q-\frac{1}{\gamma})$ a $\downarrow$-step.
Since the recurrence terminates only when $p < \basep$,
every path that hits the base case of the recurrence tree must use $\ceil{p - \basep}$ many $\nwarrow$-steps.
Since all the $\nwarrow$-steps preserve the sum $p + q$,
and each $\downarrow$-step decreases the sum by $\gamma^{-1}$,
the total number of $\downarrow$-steps used on any path before reaching $p < \basep$ is at most $\gamma \cdot (\ceil{p - \basep} + \ceil{q - \baseq}) + O(1)$.
When $q < \baseq$, a $\nwarrow$ step must be taken.
Therefore, the number of ways to reach the base case (\ie $p < \basep$) is upper-bounded by the number of $\set{\nwarrow, \downarrow}$-paths on the grid defined by $\set{(p, q) : p \ge \basep, q \ge 0}$
starting from $(\log_\lambda T, \log_\lambda k)$.
If we let $x = \ceil{\log_\lambda T - \basep}$ and $y = \ceil{\log_\lambda T + \log_\lambda k - \baseq}$,
the number of such paths is asymptotically upper-bounded by $\binom{x + \gamma y}{x}$. 
The total verifier work satisfies
\begin{align*}
    \Vtime(T, k) < \binom{x + \gamma y}{x} \cdot ((\basek + \baset \cdot \lambda^2 S) + \poly(d, \lambda, S)).
\end{align*}
\paragraph{Choosing $\gamma$} 
For a fixed choice of $\basek, \baset, S$, 
with the observation that $d = \lambda^{\gamma^{-1}}$,
\[
\Vtime(T, k) > \max\left(\binom{x + \gamma y}x, \lambda^{O(\gamma^{-1}) + 2}\right).
\]
An optimal choice of $\gamma$ balances the two terms in the $\max$.
By Stirling's approximation,
\begin{align*}
    \log \binom{x + \gamma y}{x} &\approx O\left(x \log\left(1 + \frac{\gamma y}{x}\right) + \gamma y \log\left(1 + \frac{x}{\gamma y}\right)\right)\\
    \log(\lambda^{O(\gamma^{-1})+2}) &= O((\gamma^{-1} + 2) \log \lambda).
\end{align*}
Let $X \coloneqq \log T$, $Y \coloneqq \log T + \log k$.
Setting the above equal gives
\begin{align*}
    (1 + 2\gamma)\log^2 \lambda = \gamma X \log \left(1 + \frac{\gamma Y}X\right) + \gamma^2 Y \log \left(1 + \frac{X}{\gamma Y}\right).
\end{align*}
Since $\log \Vtime(T, k) > (\gamma^{-1} + 2) \log \lambda$, we have
$$\begin{aligned}
    \log^2 (\Vtime(T, k)) &\gtrsim \frac{(\gamma^{-1} + 2)^2}{1 + 2\gamma} \left[ \gamma X \log\left(1+\frac{\gamma Y}{X}\right) + \gamma^2 Y \log\left(1 + \frac{X}{\gamma Y}\right) \right] \\
    &= (1 + 2\gamma) \left[ \frac{X}{\gamma} \log\left(1+\frac{\gamma Y}{X}\right) + Y \log\left(1 + \frac{X}{\gamma Y}\right) \right] \coloneqq f(\gamma).
\end{aligned}$$
\begin{fact}
    $f(\gamma)$ is minimized when $\gamma = \Theta(1)$.
\end{fact}
\begin{proof}
    If $\gamma = \omega(1)$, then
    \begin{align*}
        Y \log\left(1 + \frac{X}{\gamma Y}\right) = \Omega(1),
    \end{align*}
    because $Y \ge X$ and the approximation $\log(1 + u) = \Theta(u)$ for $u \to 0$.
    Therefore, $f(\gamma) \to \infty$.

    Moreover,
    let $g(\gamma) =\left[ \frac{X}{\gamma} \log\left(1+\frac{\gamma Y}{X}\right) + Y \log\left(1 + \frac{X}{\gamma Y}\right) \right]$.
    We have $f(\gamma) \ge g(\gamma)$.
    Computing $g'(\gamma)$,
    \begin{align*}
    g'(\gamma) &= -\frac{X}{\gamma^2} \log\left(1+\frac{\gamma Y}{X}\right) + \frac{X}{\gamma}\frac{Y/X}{1+\gamma Y/X} + Y\frac{-X/(\gamma^2 Y)}{1+X/(\gamma Y)} \\
    &= -\frac{X}{\gamma^2} \log\left(1+\frac{\gamma Y}{X}\right) + \frac{XY}{\gamma X + \gamma^2 Y} - \frac{XY}{\gamma^2 Y + \gamma X} \\
    &= -\frac{X}{\gamma^2} \log\left(1+\frac{\gamma Y}{X}\right) < 0.
    \end{align*}
    Therefore,
    $g$ is decreasing in $\gamma$.
    This implies that $f$ is minimized when $\gamma = \Theta(1)$.
\end{proof}
\paragraph{Selecting $\lambda$}
Given that $\gamma = \Theta(1)$,
$\binom{x + \gamma y}{x} = \binom{x + y}{x}^{\Theta(1)}$.
We select a $\lambda$ that balances the terms in $f(\gamma)$. 
Indeed,
if we set $\lambda = \Lambda(T, k) = \ceil{2^{\sqrt{\log \binom{X + Y}{X}}}}$,
then
\begin{align*}
    f(\gamma) &= \Theta(\log^2 \lambda) = \Theta\left(\log \binom{X + Y}{X}\right),\\
    \Vtime(T, k) &\ge \max\left(\binom{x + \gamma y}{x}, \poly(\lambda^{\gamma^{-1} + 2})\right) \\
    &= 2^{\Theta(\sqrt{f(\gamma)})} = 2^{\Theta\left(\sqrt{\log \binom{X + Y}{X}}\right)} = \lambda^{\Theta(1)},
\end{align*}
By the recurrence relation's upper bound, we also have
\begin{align*}
    \Vtime(T, k) \le \binom{x + \gamma y}{x}\cdot (\basek + \baset \cdot \poly(S, \lambda)) = \lambda^{\Theta(1)} \cdot (\basek + \baset\cdot \poly(S)).
\end{align*}
\paragraph{Selecting $\basek$, $\baset$}
Given that the optimal $\lambda \ge 2^{\sqrt{\log \binom{X + Y}{X}}} = \Omega(\log T + \log k)$,
the following claim can be applied.
\begin{claim}
    Suppose $\lambda \ge x + \gamma y\in O(\log T + \log k)$, then $\binom{x + \gamma y}{x}\cdot (\basek + \baset \cdot \poly(S, \lambda))$ is minimized when $\baset = O(\lambda)$ and $\basek = O(d^2)$.
    \label{clm:base-params}
\end{claim}
\begin{proof}[Proof Sketch]
    Consider the perturbation 
    \begin{align*}
    \begin{cases}
    \baset &\mapsto \baset \cdot \lambda,\\
    x &\mapsto x - 1.
    \end{cases}
    \end{align*}
    The ratio of the upper bounds is 
    (ignoring terms unchanged in the multiplicand) 
    \begin{align*}
        \frac{\binom{x - 1 + \gamma y}{x - 1}}{\binom{x + \gamma y}{x}} \cdot \lambda = \left(\frac{\lambda}{x + \gamma y}\right)\cdot x> 1.
    \end{align*}
    Therefore,
    we must set $\baset$ to be as small as possible, which is $O(\lambda)$.
    A similar argument holds for $\basek$.
\end{proof}

\end{document}